\documentclass[nomarks,final]{dmtcs-episciences}

\usepackage[utf8]{inputenc}
\usepackage{subfigure}
\usepackage[table]{xcolor}
\usepackage{graphicx}

\usepackage{hyperref}
\usepackage{amsthm}
\usepackage{xspace}
\usepackage{amssymb}
\usepackage{amsmath}
\usepackage{cleveref}
\usepackage{enumitem}
\usepackage{tikz}
\usepackage{tikz-cd}
\usepackage{todonotes}
\usepackage{placeins}
\usetikzlibrary{calc}
\usetikzlibrary{arrows}
\usetikzlibrary{decorations.pathmorphing}
\usepackage{tabularx}

\newcommand{\thr}{\ensuremath{t}} 

\newcommand{\TSS}{\textsc{Target Set Selection}\xspace}
\newcommand{\TSSshort}{\textsc{TSS}\xspace}
\newcommand{\VC}{\textsc{Vertex Cover}\xspace}

\newcommand{\VCshort}{\textsc{VC}\xspace}
\newcommand{\IS}{\textsc{Independent Set}\xspace}
\newcommand{\ISshort}{\textsc{IS}\xspace}
\newcommand{\SAT}{\textsc{Sat}\xspace}
\newcommand{\tSAT}{\textsc{3-Sat}\xspace}
\newcommand{\pSAT}{\textsc{Planar Sat}\xspace}
\newcommand{\ptSAT}{\textsc{Planar 3-Sat}\xspace}
\newcommand{\rptSAT}{\textsc{Restricted Planar $3$-Sat}\xspace}

\renewcommand{\P}{\textsf{P}\xspace}
\newcommand{\NP}{\textsf{NP}\xspace}
\newcommand{\NPh}{\NP-hard\xspace}
\newcommand{\NPhness}{\NP-hardness\xspace}
\newcommand{\NPc}{\NP-complete\xspace}

\newcommand{\N}{\ensuremath{\mathbb{N}}}

\newcommand{\Oh}[1]{\ensuremath{{\mathcal{O}(#1)}}}
\newcommand{\oh}[1]{\ensuremath{{o(#1)}}}

\usepackage{etoolbox}

\theoremstyle{plain}
\newtheorem{theorem}{Theorem}
\newtheorem*{theorem*}{Theorem}
\newtheorem{lemma}[theorem]{Lemma}
\newtheorem*{lemma*}{Lemma}
\newtheorem{corollary}[theorem]{Corollary}
\newtheorem*{corollary*}{Corollary}
\newtheorem{observation}[theorem]{Observation}
\newtheorem*{observation*}{Observation}

\newtheorem*{fact*}{Fact}

\newtheorem*{conjecture*}{Conjecture}

\theoremstyle{definition}
\newtheorem{definition}[theorem]{Definition}
\newtheorem{remark}[theorem]{Remark}
\newtheorem*{remark*}{Remark}

\newtheorem*{example*}{Example}
\newtheorem{claim}[theorem]{Claim}
\newtheorem*{claim*}{Claim}
\newtheorem*{question*}{Question}

\crefname{claim}{Claim}{Claims}
\crefname{observation}{Observation}{Observations}
\crefname{corollary}{Corollary}{Corollaries}
\Crefname{figure}{Fig.}{Figs.}
\Crefname{tabular}{Tab.}{Tabs.}
\Crefname{tabularx}{Tab.}{Tabs.}
\Crefname{table}{Tab.}{Tabs.}

\newcommand{\mycbox}[1]{\tikz{\path[draw=#1,fill=#1] (0,0) rectangle (0.2cm,0.2cm);}}
\newcommand{\myfbox}[2]{\tikz{\path[draw=#2,fill=#1,line width=0.07cm] (0,0) rectangle (0.2cm,0.2cm);}}

\newcommand\mlnode[1]{\fbox{\begin{tabular}{@{}c@{}}#1\end{tabular}}}

\makeatletter
\def\blfootnote{\gdef\@thefnmark{}\@footnotetext}
\makeatother

\usepackage[round]{natbib}

\author[Michal Dvořák et. al]{Michal Dvořák\affiliationmark{1}
  \and Dušan Knop\affiliationmark{1}
  \and Šimon Schierreich\affiliationmark{1}}
\title[On the Complexity of Target Set Selection]{On the Complexity of Target Set Selection in Simple Geometric Networks}
\affiliation{
  Department of Theoretical Computer Science, Faculty of Information Technology, Czech Technical University in Prague, Prague, Czech Republic
}
\begin{document}



\publicationdata{vol. 26:2}{2024}{11}{10.46298/dmtcs.11591}{2023-07-17; 2023-07-17; 2024-03-03; 2024-06-11}{2024-06-27}


\maketitle

\begin{abstract}
  \medskip
  We study the following model of disease spread in a social network. At first, all individuals are either \emph{infected} or \emph{healthy}. Next, in discrete rounds, the disease spreads in the network from infected to healthy individuals such that a healthy individual gets infected if and only if a sufficient number of its direct neighbors are already infected.
  
  We represent the social network as a graph. Inspired by the real-world restrictions in the recent epidemic, especially by social and physical distancing requirements, we restrict ourselves to networks that can be represented as geometric intersection graphs.
  
  We show that finding a minimal vertex set of initially infected individuals to spread the disease in the whole network is computationally hard, already on unit disk graphs. Hence, to provide some algorithmic results, we focus ourselves on simpler geometric graph classes, such as interval graphs and grid graphs.
\end{abstract}

%
%

\section{Introduction}

\blfootnote{This is an extended and revised version of a preliminary conference report that was presented at the 18th International Workshop on Algorithms and Models for the Web Graph, WAW~'23~\citep{DvorakKS23}.}

In this work, we study the following deterministic model of disease spread. We are given a social network represented as a simple, undirected graph~${G=(V,E)}$, a threshold function $t\colon V\to\N$ that associates each agent with her \emph{immunity} (or \emph{threshold}), and a \emph{budget} $k\in\N$. Our goal is to select a group $S\subseteq V$, $|S| \leq k$, of initially infected agents (a \emph{target set}) such that all agents get infected by the following activation process:
\begin{align*}
S_0 &= S,\\
S_{i} &= S_{i-1}\cup\{v\in V \mid \thr(v) \leq |N(v) \cap S_{i-1}|\}.
\end{align*}

In other words, the disease spreads in discrete rounds. A healthy agent $v$ becomes infected if the number of neighbors already infected reaches the agent's immunity value $\thr(v)$. We note that once an agent is infected, she remains in this state for the rest of the process. We shall also refer to infected agent as \emph{activated} or \emph{active} agent.

\cite{DreyerR2009} studied a similar model under the name \textsc{Irreversible $k$-threshold Process}. Unlike our setting, in their work, the immunity value is the same for all agents. Therefore, the presented model is more general and corresponds, in fact, to the \TSS problem (\TSSshort for short) where thresholds can be agent-specific.

The \TSS problem was introduced by \cite{DomingosR2002} in the context of viral marketing on social networks. \cite{KempeKT2015} later refined the problem in terms of thresholds, which is the model we follow in this work, and showed that the problem is \NPh.

The \NPhness of \TSSshort comes from the observation that for $t(v) = \deg(v)$ the problem is equivalent to the \VC problem. This setting of the threshold function is referred to as \emph{unanimous thresholds}.

The first way to tackle the complexity of the problem was aimed at the threshold function. However, \cite{Chen2009} showed that \TSSshort remains \NPh even if all thresholds are at most two, which extends the previous result of \cite{DreyerR2009} who showed that the problem is \NPh even if all thresholds are bounded by a constant $c\geq 3$. \NPhness for majority thresholds (for every $v\in V$ we have $\thr(v) = \lceil\deg(v)/2\rceil$) is due to \cite{Peleg1996}.

It is easy to see that the \TSSshort problem is solvable in polynomial time if the underlying graph has diameter one, that is, it is a complete graph~\citep{NichterleinNUW2013,ReddyKR2010}. \cite{Chen2009} showed that the problem remains polynomial-time solvable when the underlying graph is a tree. Later, \cite{ChiangHBWY2013} proposed linear-time algorithms for block-cactus graphs, chordal graphs with all thresholds at most two, and Hamming graphs with all thresholds equal to two. \cite{BessyEPR2019} showed that the \TSSshort problem is solvable in polynomial time on interval graphs if all thresholds are bounded by a constant. On the other hand, the problem becomes \NPh on graphs of diameter two~\citep{NichterleinNUW2013}.

The problem becomes solvable in polynomial time when the input graph is \mbox{$3$-regular} and all thresholds are equal to $2$. This setting is in fact equivalent to the \textsc{Feedback Vertex Set} problem~\citep{TakaokaU2015}, which is solvable in polynomial time on \mbox{$3$-regular} graphs~\citep{UenoYS1988}. More generally, setting where $t(v)=\deg(v)-1$ for all vertices $v$ is equivalent to the \textsc{Feedback Vertex Set} problem. Even more generally, \TSSshort with threshold function $t(v) = \deg(v) - d$ for all vertices $v$ and some $d\geq 0$ is equivalent to the \mbox{\textsc{$d$-Degenerate Vertex Deletion}} problem, where the task is to delete a set of vertices of size at most $k$ such that the remaining graph is $d$-degenerate. The setting with thresholds equal to $2$ was further examined by \cite{KynclLV2017}. They extended the tractability result for \TSSshort with thresholds equal to $2$ when the input graph has degree at most $3$ and showed \NPhness when the input graph has maximum degree at most $4$. The study of fast exponential algorithms for \TSSshort with constant thresholds is due to~\cite{BliznetsS2023}.

Restriction of the underlying graph structure was further investigated by \cite{BenZwiHLN2011}. They gave an algorithm running in $n^\Oh{\omega}$ time for networks with $n$ vertices and \emph{tree-width}
bounded by $\omega$, and showed that, under reasonable theoretical assumptions, there is no algorithm for \TSSshort running in $n^\oh{\sqrt{w}}$ time. The parameterized complexity perspective, initiated by \cite{BenZwiHLN2011}, was later used in multiple subsequent works~\citep{Mathieson2010,NichterleinNUW2013,ChopinNNW2014,Hartmann2018,DvorakKT2022,KnopSS2022,BanerjeeMP22,ChuL2023}.

Very recently, \cite{Schierreich2023} initiated the study of \TSSshort in dynamic environments such as temporal graphs. As social networks naturally change over time, the model with temporal graphs captures more realistically their dynamics. This variant was further investigated by~\cite{DeligkasEGS2023}. 

Finally, \cite{CicaleseCGMV2014} proposed a study of the \TSSshort problem where the process must stabilize within a prescribed number of rounds. They gave a polynomial-time algorithm for graphs of bounded \emph{clique-width} 
and a linear-time algorithm for trees. Somewhat opposite is the goal in the variant recently introduced by \cite{KeilerLMSS2023}, who studies the complexity of finding processes that lasts at least $k$ rounds.

\medskip

Inspired by the actual restrictions in the recent epidemic, especially by social and physical distancing requirements, we study the \TSS problem restricted to instances where the underlying graph is a \emph{(unit) disk graph}.

Unit disk graphs were initially used as a natural model for the topology of \mbox{ad-hoc} wireless communication networks~\citep{HusonS1995}. For a given graph, it is \NPh to recognize whether the graph is a unit disk graph~\citep{BreuK1993,HlinenyK2001,KangM2012}. On the other hand, many computationally hard problems, such as \IS or \textsc{Colouring}, can be efficiently approximated for this graph class~\citep{Matsui2000}. \textsc{Clique} can be solved even in polynomial time if the disk representation is given as part of the input~\citep{ClarkCJ1990}.

In our case, the disk representation models two different situations. In the first situation, the disk represents the distances that individuals must keep. In the second case, the disk represents the area in which the disease is spread by an infected individual.

As the \TSS problem is notoriously hard from both exact computation and approximation point of view, it is natural to ask whether any positive result can be given if we restrict \TSSshort to instances where the underlying graph is a unit disk graph or if we need to restrict ourselves to even simpler graph classes.

%
%
\subsubsection*{Our Contribution}
In this paper, we further extend the complexity picture of \TSS problem and show that it is computationally hard even on very simple geometric graph classes. In particular, we show that \TSS is \NPc in the class of unit disk graphs even if the threshold function is bounded by a constant $c\geq 2$, is equal to majority, or is unanimous. Hence, we focus on the study of grid graphs, which is a subclass of unit disk graphs. For grid graphs, we show that \TSS is \NPc for thresholds at most $2$ even when the maximum degree is at most~$3$. As a corollary, we show an \NPhness result for \TSSshort with majority thresholds for the class of grid graphs. Lastly, we give \NPhness results for the case when thresholds are set to a constant. We show that \TSS is \NPc in the class of unit disk graphs when $t(v)=2$ and $\Delta G \leq 4$. Note that our results for constant thresholds establish a clear dichotomy between tractable and intractable subclasses of geometric intersection graph classes, as \TSSshort is known to be solvable in polynomial time on interval graphs with constant thresholds~\citep{BessyEPR2019} and unanimous thresholds~\citep{Farber1982} (i.e. the \VC problem). As a byproduct of our theorems we also obtain the full complexity picture of \TSS in the class of planar graphs. We emphasize that most of our reductions yield a graph with very small maximum degree and unless $\P=\NP$ the maxium degree cannot be further reduced. For a graphical overview of our results we refer the reader to \Cref{tbl:our_contribution} and \Cref{fig:classes}.

\subsubsection*{Paper Organization}
The remainder of this paper is organized as follows. In \Cref{sec:preliminaries}, we introduce all the definitions and notation used throughout the paper. In \Cref{sec:unanimous_thresholds}, we show the hardness and algorithmic results for \TSS restricted to unanimous thresholds. \Cref{sec:constant_thresholds} is dedicated to the variant where the maximum threshold is bounded by a constant. In \Cref{sec:majority_thresholds}, we study a variant of the \TSS problem with majority thresholds. In \Cref{sec:exact_thresholds}, we explore the difference between constant thresholds and exact thresholds and we conclude the paper with open problems and future research directions in \Cref{sec:conclusions}.

\begin{table}[ht!]
	
	{\small
		\newcolumntype{Y}{>{\centering\arraybackslash}X}
		\renewcommand{\arraystretch}{1.5}
		\begin{tabularx}{\textwidth}{|p{2.5cm}|X|X|X|X|}\hline
			& constant & majority & unanimous & unrestricted \\
			\hline
			
			\textbf{interval graphs}   & 
			\P~\citep{BessyEPR2019} & 
			? & 
			\P~\citep{Farber1982} & 
			? \\\hline
			
			\textbf{grid graphs}       & 
			\NP-c (Thm.~\ref{thm:tss_npc_constant_grid}) & 
			\NP-c (Cor.~\ref{cor:tss_npc_majority_grid}) & 
			\P~\citep{ClarkCJ1990} & 
			\NP-c (Thm.~\ref{thm:tss_npc_constant_grid}) \\\hline
			
			\textbf{unit disk graphs}  & 
			\NP-c (Cor.~\ref{cor:tss_npc_constant_udg}) & 
			\NP-c  (Cor.~\ref{cor:tss_npc_majority_udg}) & 
			\NP-c~\citep{ClarkCJ1990} & 
			\NP-c~\citep{ClarkCJ1990} \\\hline
			
			
			\textbf{planar graphs} &
			\NP-c (Thm. \ref{thm:tss:planar:npc:const}) &
			\NP-c (Cor. \ref{cor:tss_npc_majority_planar})&
			\NP-c~\citep{FleischnerSS2010}&
			\NP-c~\citep{FleischnerSS2010}\\
			\hline
		\end{tabularx}
	}
	\caption{Overview of our results. The first row contains individual restrictions of the threshold function, and the first column contains assumed graph classes. In the table, ``\NP-c'' stands for ``\NPc'', ``\P'' stands for polynomial-time solvable cases, and ``?'' indicates an open question. The new results from this paper are marked with a reference to the appropriate statement.}
	\label{tbl:our_contribution}
\end{table}

\begin{figure}[ht!]
	\centering
	\begin{tikzpicture}
		\node (a) at (0,0){
			$$
			\begin{tikzcd}[ampersand replacement=\&]
				\& \mlnode{DISK \\ \myfbox{red}{gray} {}\myfbox{red}{black} {}\myfbox{red}{black} {}\myfbox{red}{gray} }        \&                           \\
				\mlnode{PLANAR \\ \myfbox{red}{gray} {}\myfbox{red}{black} {}\myfbox{red}{black} {}\myfbox{red}{gray}}\arrow[ru]             \& \mlnode{UDISK \\ \myfbox{red}{gray} {}\myfbox{red}{black} {}\myfbox{red}{black} {}\myfbox{red}{gray}}\arrow[u] \& \mlnode{INT \\ \myfbox{green}{gray} {}\mycbox{green}{} {}\mycbox{yellow}{} {}\mycbox{yellow}{}} \arrow[lu]            \\
				\mlnode{GRID \\ \myfbox{green}{gray} {}\myfbox{red}{black} {}\myfbox{red}{black} {}\myfbox{red}{black}} \arrow[ru] \arrow[u] \&               \& \mlnode{UINT \\ \myfbox{green}{gray} {}\mycbox{green}{} {}\mycbox{yellow}{} {}\mycbox{yellow}{}} \arrow[u] \arrow[lu]
			\end{tikzcd}
			$$
		};
	\end{tikzpicture}
	\caption{Overview of our main results regarding \TSS for \emph{unanimous}, \emph{constant}, \emph{majority}, and \emph{unrestricted} threshold function. Red squares correspond to \NPhness, green correspond to polynomial-time solvability and yellow indicate an open question. Squares with borders correspond to results established in this work. Black squares indicate a main result while the gray indicate a direct or trivial corollary to a previously known result. An arrow from a class $\mathcal{A}$ to class $\mathcal{B}$ corresponds to the fact that $\mathcal{A}$ is a subclass of $\mathcal{B}$. The six classes in the picture are (by row) disk graphs, planar graphs, unit disk graphs, interval graphs, grid graphs, unit interval graphs. }\label{fig:classes}
\end{figure}
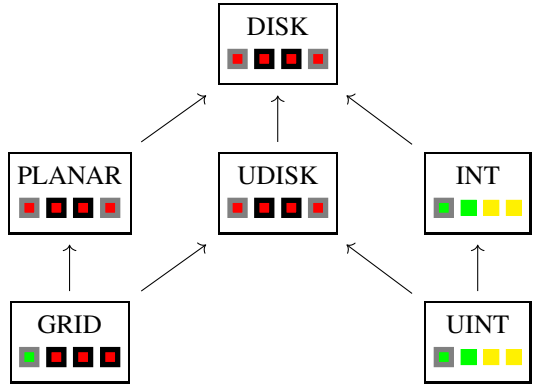

%
%
\section{Preliminaries}\label{sec:preliminaries}

For $n\in\mathbb{N}$ we denote $[n]=\{1,\ldots,n\}$, in particular, $[0]=\emptyset$. For a set $X$ and a constant $c\in\mathbb{N}$, the symbol $X^{\geq c}$ denotes the set of all $d$-tuples from $X$ where $d\geq c$.
A \emph{simple undirected graph} is a pair $G=(V,E)$, where $V$ is a set of \emph{vertices}, and $E \subseteq \binom{V}{2}$ is a set of \emph{edges}. Let $u$ and $v$ be two distinct vertices. If $\{u,v\}\in E$, then we call~$u$ a~\emph{neighbor} of~$v$ and vice versa. We denote the \emph{open neighborhood} of a~vertex~$v$ by $N(v)$ and $|N(v)|=\deg(v)$ is the \emph{degree} of the vertex~$v$. For $A\subseteq V$ the open neighborhood of $A$ is $N(A)=\bigcup_{v\in A}N(v)$. A vertex of degree~$1$ is called a \emph{leaf}. The \emph{maximum degree} of a graph $G$ is denoted by $\Delta G$.
Let~$r\in\N$ be a constant. We say that a graph $G$ is \emph{$r$-regular} if every vertex $v\in V$ has degree exactly $r$. A graph is \emph{regular}, if it is $r$-regular for some constant $r\in\N$.

\begin{definition}[Unit disk graph]
	A graph $G=(V,E)$ with $V=\{v_1,\ldots,v_n\}$ is a \emph{disk graph} if there exists a collection $\mathcal{D}=(D_1,\ldots,D_n)$ of $n$ closed disks in the Euclidean plane such that $\{v_i,v_j\}\in E$ if and only if $D_i \cap D_j \neq \emptyset$. If all disks $D_i\in\mathcal{D}$ have same diameter, we call the graph a \emph{unit disk graph}.
\end{definition}

Let $G$ and $H$ be two graphs. The \emph{Cartesian product} of the graphs $G$ and $H$ is a graph $G\square H$ such that $V(G\square H) = V(G)\times V(H)$ and $\{(u,u'),(v,v')\}$ is an edge if and only if either $u=v$ and $u'$ is a neighbor of $v'$ in $H$, or~$u'=v'$ and~$u$ is a neighbor of $v$ in $G$.

\begin{definition}[Grid graph]
	An $n\times m$ \emph{grid} is the Cartesian product of the path graphs $P_n$ and $P_m$. A graph $G$ is a \emph{grid graph} if and only if it is an induced subgraph of a grid.
\end{definition}

\begin{lemma}\label{lem:grids_are_udg}
	Every grid graph is a unit disk graph.
\end{lemma}
\begin{proof}
	Let $G$ be a grid graph. By definition, it is an induced subgraph of the cartesian product $P_n\square P_m$. Embed $G$ into the integer grid $[n]\times [m]\subseteq \mathbb{Z}\times \mathbb{Z}$ in the obvious way. Unit disk representation of $G$ is as follows. For each vertex $v\in V(G)$ place a disk with diameter $1$ centered at the corresponding grid point. Notice that since $G$ is an induced subgraph of a grid, all adjacencies are preserved and on the other hand every adjacency in the unit disk representation corresponds to an adjacency in the original grid graph~$G$.
\end{proof}

Let $G=(V,E)$ be a graph. We call a set $C\subseteq V$ a \emph{vertex cover} of~$G$, if at least one end of each edge is a member of $C$. In the \VC (\VCshort for short) problem, we are given a graph $G$ and an integer $k\in\N$, and our goal is to decide whether there is a vertex cover $C$ of size at most $k$. A set $I\subseteq V$ is called an \emph{independent set}, if for every pair of distinct vertices $u,v\in I$ there is no edge connecting~$u$ and $v$. In the \IS problem (\ISshort for short), we are given a graph $G$ and an integer $k\in\N$, and our goal is to decide whether there is an independent set $I$ of size at least $k$.

Some of our \NPhness reductions come from a variant of the \SAT problem. In the \SAT problem, we are given a propositional formula $\varphi$ in conjunctive normal form (CNF) over the set $\operatorname{Var}(\varphi)$ of \emph{variables}. The set of clauses is denoted by $\mathcal{C}(\varphi)$. Our goal is to decide whether there is a truth assignment $\pi\colon \operatorname{Var}(\varphi)\to\{0,1\}$ that satisfies~$\varphi$. In \tSAT, all clauses are restricted to be of size at most $3$. An \emph{incidence graph} for formula $\varphi$ is a bipartite graph $G_\varphi$ with vertices $v_C$ for each clause $C\in \mathcal{C}(\varphi)$ and $v_x$ for each variable $x\in \operatorname{Var}(\varphi)$ and there is an edge between $x\in \operatorname{Var}(\varphi)$ and $C\in\mathcal{C}(\varphi)$ if and only if the variable $x$ occurs in the clause~$C$. A variant of \SAT where $G_\varphi$ is planar is called \pSAT. \ptSAT is combination of the two settings mentioned above. In \rptSAT it is further assumed that each variable $x_i$ occurs exactly $3$ times -- twice as a positive literal, and once as a negative literal. This variant of \SAT is \NPh~\cite[pf. Theorem 2a]{DahlhausJPSY1994}

%
%
\section{Unanimous Thresholds}\label{sec:unanimous_thresholds}

In this section, we study the special case of the \TSS problem where the thresholds are unanimous. Recall that this means that for every vertex $v$ we have $t(v) = \deg (v)$. All our results strongly rely on the following easy-to-see and well-known equivalence between the \TSS problem with unanimous thresholds and the \VC problem. Proof of this can be found, for example, in the work of \cite{Chen2009}.

\begin{lemma}\label{lem:tss:unanimous:vc:eq}
	The \TSS problem with unanimous thresholds is equivalent to the \VC problem.
\end{lemma}
It is not hard to see that \TSSshort is indeed in \NP. A valid target set is a valid \NP certificate. In fact, all \NPh problems we deal with in this work are trivially in \NP. We always state \NP-completeness, however, we omit the part with \NP containment because it is trivial.

It is known that the \VC problem is \NPc on unit disk graphs~\citep{ClarkCJ1990}, $3$-regular planar graphs~\citep{FleischnerSS2010}, and polynomial-time solvable on interval graphs~\citep{Farber1982} and bipartite graphs by using maximum matching~\citep{Konig1931,ChuzhoyK2024}. Note that grid graphs are bipartite. By \Cref{lem:tss:unanimous:vc:eq}, the same tractability results hold for \TSS with unanimous thresholds.

%
%
\section{Constant Thresholds}\label{sec:constant_thresholds}

The \TSS problem seems to be intractable on unit disk graphs when the threshold function is unrestricted or unanimous. It is now natural to ask whether the problem remains \NPh even under some natural restrictions of the threshold function.

In this section, we show that \TSSshort is \NPc when the underlying graph is a~unit disk graph and all thresholds are bounded by a constant $c\geq 2$. Note that the case where all thresholds are at most $1$, is trivial. It is sufficient (and necessary) to choose one vertex per connected component without a vertex of threshold $0$ of the input graph. This can be accomplished in linear time.

Before delving into the proofs, we give a historical remark. If the input graph is unrestricted, the \NPhness of \TSSshort for constant thresholds is implied by the hardness of \VC and \IS in the class of $k$-regular graphs for $k\geq 3$ as noted by~\cite{DreyerR2009}. Indeed, this can be also regarded as the unanimous threshold setting. However, this is only applicable for thresholds bounded by $c\geq 3$. For thresholds bounded by $2$ in general graphs the hardness is either implied by the inapproximability result of~\cite{Chen2009} or is achieved via involved reduction from the \SAT problem~\citep{Centeno2011, KynclLV2017}. Although we eventually show hardness even for thresholds bounded by $2$ (by reduction from a special variant of \SAT), this section shows that the very same approach could be used even in the class of unit disk graphs. 

In what follows we will employ reductions from problems involving planar graphs. To effectively apply such reductions, it will be essential to have some kind of 'nice' representation of the planar graph. One such representation useful for our purposes is the so-called \emph{rectilinear embedding}. This planar embedding ultimately helps us with further reductions to the class of grid graphs.

\begin{definition}
	Given a planar graph $G=(V,E)$, a \emph{rectilinear embedding~(of~$G$)} is a planar drawing of $G$ such that vertices occupy integer coordinates, and all edges are made of (possibly more) line segments of the form $x=i$ or $y=j$ for some integers $i,j$ (i.e., the line segments are parallel to the coordinate axes).
\end{definition}

Formally, a rectilinear embedding of a planar graph $G=(V,E)$ is a pair of mappings $(\mathcal{E}^V,\mathcal{E}^E)$ where $\mathcal{E}^V\colon V\rightarrow \mathbb{Z}\times \mathbb{Z}$ and
$\mathcal{E}^E\colon E\rightarrow (\mathbb{Z}\times \mathbb{Z})^{\geq 2}$ and the following conditions must hold:
\begin{enumerate}[label=\roman*)]
	\item The mapping $\mathcal{E}^V$ is injective.
	\item For any edge $e$, the tuple $\mathcal{E}^E(e)=(p_1,p_2,\ldots, p_g)$ induces a simple polygonal chain and for all $i\in[g-1]$ the points $p_i,p_{i+1}$ are adjacent grid points. More precisely, if $p_i=(x_i,y_i)$ and $p_{i+1}=(x_{i+1},y_{i+1})$, then either $x_i=x_{i+1}\wedge |y_i-y_{i+1}|=1$ or $y_i=y_{i+1}\wedge |x_i-x_{i+1}|=1$.
	\item For any edge $e=\{u,v\}$, the points $\mathcal{E}^V(u)$ and $\mathcal{E}^V(v)$ are endpoints of the polygonal chain induced by $\mathcal{E}^E(e)$.
	\item For any two distinct edges $e,f$, the simple polygonal chains induced by $\mathcal{E}^E(e)$ and $\mathcal{E}^E(f)$ are disjoint except possibly at the endpoints.
\end{enumerate}

To simplify notation, we will use only one mapping $\mathcal{E}=\mathcal{E}^V\cup\mathcal{E}^E$ to represent a rectilinear embedding $(\mathcal{E}^V,\mathcal{E}^E)$. We slightly abuse notation and identify $1$-tuples of $(\mathbb{Z}\times \mathbb{Z})^1$ with the elements of $\mathbb{Z}\times \mathbb{Z}$. In other words, $\mathcal{E}\colon V\cup E \rightarrow (\mathbb{Z}\times \mathbb{Z})^{\geq 1}$ and $\mathcal{E}|_V=\mathcal{E}^V$, $\mathcal{E}|_E=\mathcal{E}^E$. Since $\mathcal{E}^V$ is by definition injective, for a~grid point $p\in\mathcal{E}(V)\subseteq\mathbb{Z}\times\mathbb{Z}$, we let $\mathcal{E}^{-1}(p)$ denote the vertex $v\in V$ such that $\mathcal{E}(v)=p$. The \emph{area} of a rectilinear embedding~$\mathcal{E}$, denoted $\operatorname{Area}(\mathcal{E})$, is the minimal area of an axes-parallel closed rectangle $R$ such that the embedding is contained in $R$. We utilize the following theorem of Valiant which establishes a~sufficient condition for existence of rectilinear embedding for a planar graph $G$.

\begin{theorem}[\cite{Valiant1981}]\label{thm:valiant}
	Given a planar graph $G=(V,E)$ with maximum degree $\Delta G \leq 4$, there exists a rectilinear embedding $\mathcal{E}$ of $G$ satisfying $\operatorname{Area}(\mathcal{E})\leq\Oh{|V|^2}$. Moreover, $\mathcal{E}$ can be computed in polynomial time with respect to the size of $G$.	
\end{theorem}

\subsection{Thresholds bounded by 3 in unit disk graphs}
We begin with an auxiliary result showing \NPhness of the \IS problem in \mbox{$3$-regular} and $4$-regular unit disk graphs.

\begin{theorem}\label{thm:is:npc:regular}
	\IS is \NPc even if the underlying graph is an~$r$-regular unit disk graph, where $r\in \{3,4\}$.
\end{theorem}

We divide the proof of \Cref{thm:is:npc:regular} into two parts. First we explain the construction of the reduction. In the second part we explain how to represent the resulting graph as a unit disk graph.

\subsection*{Construction}
The reduction is from the \IS problem on $r$-regular planar graphs. Let $(G,k)$ be an instance of the \IS problem where $G=(V,E)$ is a planar $r$-regular graph. We construct instance $(G',k')$ as follows. Start with $G$ and subdivide each edge $e=\{u,v\}\in E$ exactly $6q_e$ times, creating a path $ux_1^ex_2^e\dots x_{6q_e}^ev$. The number $q_e$, which depends on the edge~$e$, will be explained later. For the construction, it is only important that the number of subdivisions is a multiple of $6$. Next, for all $i\in [2q_e]$, replace each vertex $x_{3i-1}^e$ with a clique $K_{r-1}$ and connect all its neighbors to the clique. In other words, create $r-2$ additional copies of the vertex $x_{3i-1}^e$ and connect these copies into a complete graph (independently for each $i$). Let $X_e$ denote the set of vertices created by subdividing the edge $e$ (including the copies of all vertices $x_{3i-1}^e$) (see \Cref{fig:edge_subdivision}). Let $G'$ denote the resulting graph. Note that $G'$ is $r$-regular. To finish the construction, we set $k'=k+\sum_{e\in E}3q_e$.

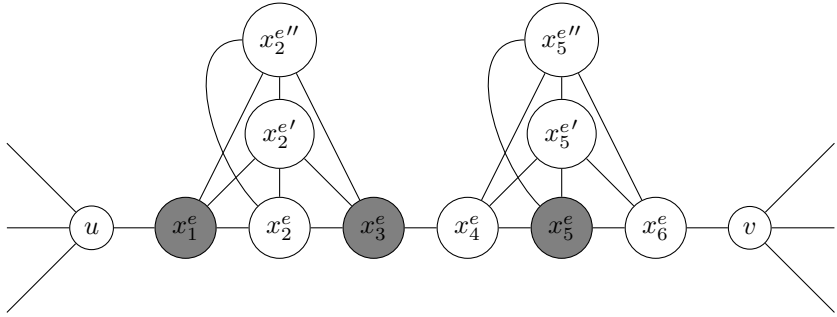
\begin{figure}[ht]
	\centering
	\begin{tikzpicture}[node distance=1.25cm]
	\node[draw,circle](u) {$u$};
	\node[draw,circle,fill=gray](x1)[right of=u] {$x_1^e$};
	\node[draw,circle](x2)[right of=x1] {$x_2^e$};
	\node[draw,circle](x22)[above of=x2] {${x_2^e}'$};
	\node[draw,circle](x222)[above of=x22] {${x_2^e}''$};
	\node[draw,circle,fill=gray](x3)[right of=x2] {$x_3^e$};
	\node[draw,circle](x4)[right of=x3] {$x_4^e$};
	\node[draw,circle,fill=gray](x5)[right of=x4] {$x_5^e$};
	\node[draw,circle](x55)[above of=x5] {${x_5^e}'$};
	\node[draw,circle](x555)[above of=x55] {${x_5^e}''$};
	\node[draw,circle](x6)[right of=x5] {$x_6^e$};
	\node[draw,circle](v)[right of=x6] {$v$};
	
	\draw (u) -- (x1);
	\draw (x1) -- (x2);
	\draw (x2) -- (x3);
	\draw (x3) -- (x4);
	\draw (x4) -- (x5);
	\draw (x5) -- (x6);
	\draw (x6) -- (v);
	\draw (x2) -- (x22);
	\draw (x22) -- (x222);
	\draw (x1) -- (x22);
	\draw (x1) -- (x222);
	\draw (x3) -- (x22);
	\draw (x3) -- (x222);
	\draw (x222) to [out=180,in=135] (x2);
	
	\draw (x5) -- (x55);
	\draw (x55) -- (x555);
	\draw (x4) -- (x55);
	\draw (x4) -- (x555);
	\draw (x6) -- (x55);
	\draw (x6) -- (x555);
	\draw (x555) to [out=180,in=135] (x5);
	
	\node[](vv)[right of=v]{};
	\node[](vva)[above of=vv]{};
	\node[](vvb)[below of=vv]{};
	\draw (v) --(vv);
	\draw (v) --(vva);	
	\draw (v) --(vvb);	
	
	\node[](uu)[left of=u]{};
	\node[](uua)[above of=uu]{};
	\node[](uub)[below of=uu]{};
	\draw (u) --(uu);
	\draw (u) --(uua);	
	\draw (u) --(uub);

	\end{tikzpicture}
	\caption{An example subdivision of the edge $e=\{u,v\}$ in the case of $4$-regular graphs. In this case, $q_{e}=1$ and  $X_e=\{x_1^e,x_2^e,x_3^e,x_4^e,x_5^e,x_6^e,{x_2^e}',{x_2^e}'',{x_5^e}',{x_5^e}''\}$. The half-edges going from $u$ and $v$ symbolize the rest of the graph. The filled vertices showcase the situation from Case 1 in proof of \Cref{lem:is:npc:onedir}, as $u\notin I_{\ell - 1}$, we add $x_1^e$, $x_3^e$ and $x_5^e$  to $I_\ell$ and the set $I_\ell$ remains independent. }
	\label{fig:edge_subdivision}
\end{figure}

We now establish the equivalence of the instances $(G,k)$ and $(G',k')$. This is the content of \Cref{lem:is:npc:onedir,lem:is:npc:twodir}.

\begin{claim}\label{lem:is:npc:onedir}
	If $(G,k)$ is a \emph{yes}-instance of \IS, then $(G',k')$ is a \emph{yes}-instance of \IS.
\end{claim}
\begin{proof}
	Assume that $(G,k)$ is a \emph{yes}-instance and let $I$ be an independent set in $G$ of size at least $k$. We build an independent set $I'$ in $G'$ of size at least $k'$. Let $E(G)=\{e_1,e_2,\ldots,e_m\}$ be enumeration of all edges of $G$ in an arbitrary order.  We inductively build a chain $I_0\subseteq I_1\subseteq \cdots \subseteq I_m$ of independent sets in $G'$ and set $I'=I_m$. Towards this, we start with $I_0=I$ and add vertices created by subdivisions, i.e. vertices from $X_{e_1},X_{e_2},\ldots,X_{e_m}$. For all $\ell\in[m]$ it will hold that $I_{\ell}\cap V(G)=I$. Notice that for $\ell = 0$, since $I\subseteq V(G)$, the property $I_0\cap V(G)=I$ holds. For the inductive step, let $\ell \geq 1$ and let $e_\ell = \{u,v\}$. Since $I$ is independent and $I_{\ell - 1}\cap V(G)=I$ it cannot happen that $\{u,v\}\subseteq I_{\ell-1}$. In other words, at least one of $u,v$ is not in $I_{\ell -1}$. We now distinguish two not necessarily exclusive cases (if both of them apply, choose one arbitrarily):
	\begin{description}
		\item[Case 1] If $u\notin I_{\ell-1}$, we set $I_{\ell}=I_{\ell - 1}\cup \{x_{2i-1}^{e_\ell}\mid i \in [3q_e]\}$.
		\item[Case 2] If $v\notin I_{\ell-1}$, we set $I_{\ell}=I_{\ell - 1}\cup \{x_{2i}^{e_\ell}\mid i \in [3q_e]\}$.
	\end{description}
	It is straightforward to verify that for all $\ell\in[m]$ the property $I_\ell\cap V(G)=I$ holds and $I_\ell$ is independent. Since $|I_0|=|I|\geq k$ and $|I_{\ell}|=|I_{\ell-1}|+3q_{e_\ell}$ for all $\ell \in [m]$, it indeed holds that $|I'|=|I_m|=|I_0|+\sum_{e\in E(G)}3q_e\geq k + \sum_{e\in E(G)}3q_e = k'$. It follows that $(G',k')$ is a \emph{yes}-instance.
\end{proof}
\begin{claim}\label{lem:is:npc:twodir}
	If $(G',k')$ is a \emph{yes}-instance of \IS, then $(G,k)$ is a \emph{yes}-instance of \IS.
\end{claim}
\begin{proof}
	Assume that $(G',k')$ is a \emph{yes}-instance and let $I'$ be an independent set in $G'$ of size at least $k'$. We build an independent set $I$ in $G$ of size at least $k$. Let $E(G)=\{e_1,\ldots,e_m\}$ be an enumeration of all edges of~$G$ in an arbitrary order. We inductively build a chain $I_0\supseteq I_1 \supseteq \cdots \supseteq I_m$ of independent sets in~$G'$ with $I_m\subseteq V(G)$ and we set $I=I_m$. Start with $I_0=I'$. Now let $\ell \geq 1$ and assume that the set~$I_{\ell-1}$ is already built and is independent. We describe how to build $I_\ell$. Let $e_\ell = \{u,v\}$. There are two cases to consider:
	\begin{description}
		\item[Case 1] At least one of $u$ and $v$ is not in $I_{\ell-1}$. In this case, since $I_\ell$ is independent, we have $|I_{\ell-1}\cap X_{e_\ell}|\leq 3q_{e_\ell}$ by the pigeonhole principle. We set $I_{\ell} = I_{\ell-1}\setminus X_{e_\ell}$.
		\item[Case 2] Both $u$ and $v$ are in $I_{\ell-1}$. In this case, $|I_{\ell-1}\cap X_{e_\ell}|\leq 3q_{e_\ell}-1$ by the same argument as above. We set $I_{\ell}=I_{\ell-1}\setminus (X_{e_\ell}\cup \{u\})$.
	\end{description}
	The resulting set $I=I_m$ is indeed independent. To see this let $e=\{u,v\}\in E(G)$ be an arbitrary edge. Notice that at the time of processing the edge $e_\ell=e$ one of $u,v$ was already missing in $I_{\ell-1}$ {(Case~1)} or we explicitly removed $u$ from $I_{\ell-1}$ {(Case~2)}. It remains to show that $|I|\geq k$.
	Note that $|I_{\ell-1}|-|I_{\ell}|\leq 3q_{e_\ell}$ for all~$\ell\in[m]$. It follows that $|I|=|I_m|\geq |I_0|- \sum_{e\in E(G)}3q_e=k'-\sum_{e\in E(G)}3q_e=k$. Thus, $(G,k)$ is a \emph{yes}-instance.
\end{proof}

We now turn our attention to how to represent the graph $G'$ as a unit disk graph. We also explain how to compute the constants $q_e$ for each edge $e$, and we show that they can be bounded by a polynomial in the size of $G$.

\subsection*{Unit disk representation}

We let $d=\frac{1}{7}$ be the diameter of the disks in the representation. Since~$G$ is planar and in both cases (i.e.,~$G$~is~$3$-~or~$4$-regular) we have $\Delta G\leq 4$, by \Cref{thm:valiant} there is a rectilinear embedding~$\mathcal{E}$~of~$G$ of polynomial area and computable in polynomial time. We now describe how to represent the vertices of~$G'$ with disks.

First, the vertices $v\in V(G')$ corresponding to vertices of $G$ will have their disk centered at the grid point $\mathcal{E}(v)$.

We now show how to construct the subdivisions of the edges. We proceed independently for each edge~$e\in E(G)$\footnote{All variables introduced from this point should have additional superscript~$e$ to signify their dependence on the edge~$e$. Nonetheless, for the sake of readability, we omit it.}. Let $\mathcal{E}(e)=(p_1,\ldots,p_g)$. We place disks $D_2,\ldots,D_{g-1}$ centered at the points $p_2,\ldots,p_{g-1}$. Let $D_1$ and $D_g$ denote the disks corresponding to vertices $\mathcal{E}^{-1}(p_1)$ and $\mathcal{E}^{-1}(p_g)$, respectively. Our task is now to insert a certain number of disks between $D_{i},D_{i+1}$ for all $i\in [g-1]$ such that the total number of disks between $D_1$ and $D_g$ is a multiple of~$6$, that is, the total number of disks between $D_1$ and $D_g$ (excluding $D_1$ and $D_g$) should be equal to $6q_e$. Let $w_i$ denote the number of disks inserted between $D_i,D_{i+1}$. We specify the numbers $w_i$ later. The total number of disks between $D_1,\ldots, D_g$ is therefore given by $y_e=g-2+\sum_{i=1}^{g-1}w_i$. Our aim is now to choose the numbers $w_i$ such that $y_e=6q_e$.

To achieve this, we do the following. First, we learn how to insert $\ell\in \{6,7,8,9\}$ disks between a single pair of adjacent disks $D_i$ and $D_{i+1}$. We prove this in the following lemma. Note that $D_i$ and $D_{i+1}$ are centered at neighboring grid points since, by definition, $p_i,p_{i+1}$ are neighboring grid points. We simplify the scenario and assume that $D_i$ and $D_{i+1}$ are centered at $p_i=(0,0)$ and $p_{i+1}=(1,0)$, respectively. It is not hard to generalize this idea to general points $p_{i},p_{i+1}$.

\begin{lemma}\label{lem:is:npc:construction_technical_lemma1}
	Let $L$ be the line segment with endpoints $(0,0),(1,0)$ and let $\ell\in\{6,7,8,9\}$. There exist $\ell$ disks $E_1,\ldots, E_\ell$ with diameter $d=\frac{1}{7}$ and centers $s_1,\ldots, s_\ell$ all lying on $L$ such that:
	\begin{enumerate}[label=\roman*)]
		\item $s_1=(d,0)$,\label{prop1}
		\item $s_\ell=(1-d,0)$,\label{prop2}
		\item any disk $E_j$ intersects exactly its neighbors $E_{j-1}$ and $E_{j+1}$ (if they exist).\label{prop3}
	\end{enumerate}
\end{lemma}
\begin{proof}
	We prove this by construction and specify the centers of the $\ell$ disks. As all centers shall lie on the line $L$ they are of the form $s_j=(a_j,0)$. For fixed $\ell$ and $j\in [\ell]$, the numbers $a_j$ are given by the formula:
	$$
	a_j=\frac{5j+\ell-6}{7(\ell-1)}.
	$$
	It can be verified by a straightforward calculation that the properties \ref{prop1}, \ref{prop2} and \ref{prop3} hold. To verify \ref{prop3}, it is enough to check that $a_{j+1}-a_j\leq d$ and $a_{j+2}-a_j>d$ for appropriate $j$.
\end{proof}

Now we know how to insert $\ell\in \{6,7,8,9\}$ disks. We show how many disks we have to insert such that $y_e=g-2+\sum_{i=1}^{g-1}w_i$ is a multiple of $6$, given $g\geq 2$. In other words, we are now in the situation to choose the $w_i$'s, given $g\geq 2$. We prove this in the following lemma.
\begin{lemma}\label{lem:is:npc:construction_technical_lemma2}
	For any $g\geq 2$ there exist $g-1$ numbers $w_1,\ldots,w_{g-1}\in \{6,7,8,9\}$ such that
	$$
	g-2+\sum_{i=1}^{g-1}w_i=0\mod 6.
	$$ 
\end{lemma}
\begin{proof}
	We divide the proof into six cases according to the residue class of $g$ modulo $6$.
	\begin{description}
		\item[Case 1] If $g=0\mod 6$, set $w_1=8$ and $w_i=6$ for $i\in\{2,\dots,g-1\}$.
		\item[Case 2] If $g=1\mod 6$, set $w_1=7$ and $w_i=6$ for $i\in\{2,\dots,g-1\}$.
		\item[Case 3] If $g=2\mod 6$, set $w_i=6$ for all $i\in[g-1]$.
		\item[Case 4] If $g=3\mod 6$, set $w_1=9,w_2=8$ and $w_i=6$ for all $i\in\{3,\dots,g-1\}$.
		\item[Case 5] If $g=4\mod 6$, set $w_1=w_2=8$ and $w_i=6$ for all $i\in\{3,\dots,g-1\}$.
		\item[Case 6] If $g=5\mod 6$, set $w_1=9$ and $w_i=6$ for all $i\in\{2,\dots,g-1\}$.
	\end{description}
	It is a straightforward computation to verify that the chosen numbers $w_i$ work in every case. Note that $g\geq 3$ in Case 4 and 5.
\end{proof}

The formula for $q_e$ is thus given by 
$$
q_e=\frac{1}{6}y_e=\frac{1}{6}\left(g-2+\sum_{i=1}^{g-1}w_i\right).
$$
The number $g$ is given by the polygonal chain induced by $\mathcal{E}(e)$ and \Cref{lem:is:npc:construction_technical_lemma2} tells us how to choose the numbers $w_i$.

It remains to show how to represent the cliques that are replaced for the vertices $x_{3i-1}^e$ for $i\in[2q_e]$. Simply create $r-2$ additional copies of the corresponding disk in the representation.

\begin{figure}[ht]
	\centering
	\begin{tikzpicture}[scale=0.5]
		
		\begin{scope}
			\draw[fill=red,line width = 0.03cm] (-3.000000, 6.000000) circle (0.214286);
			\draw[fill=red,line width = 0.03cm] (0.000000, 0.000000) circle (0.214286);
			\draw[fill=red,line width = 0.03cm] (0.000000, 3.000000) circle (0.214286);	
			\draw[fill=red,line width = 0.03cm] (3.000000, 3.000000) circle (0.214286);

			

			\draw[line width=0.25mm] (-3-0.214286,6) -- (-4.5,6);
			\draw[line width=0.25mm] (-3,6+0.214286) -- (-3,7.5);
			\draw[line width=0.25mm] (3+0.214286,3) -- (4.5,3);
			\draw[line width=0.25mm] (0,3+0.214286) -- (0,7.5);
			
			\draw[line width=0.25mm] (0,0.214286) -- (0,3-0.214286);
			\draw[line width=0.25mm] (0.214286,0) -- (3,0);
			\draw[line width=0.25mm] (3,3-0.214286) -- (3,0);
			\draw[line width=0.25mm] (0-0.214286,0) -- (-3,0);
			\draw[line width=0.25mm] (-3,0) -- (-3,6-0.214286);
			\draw[line width=0.25mm] (0+0.214286,3) -- (3-0.214286,3);
			
			\node[] (a1) at (-2.4,6) {$a_1$};
			\node[] (a2) at (-0.6,3) {$a_2$};
			\node[] (a3) at (3,3.4) {$a_3$};
			\node[] (a4) at (0,-0.5) {$a_4$};
			
		\end{scope}
		
		\draw[->] (7,3) -- (8.5,3);

		\begin{scope}[xshift=15cm]
			\draw[fill=red,line width = 0.03cm] (-3.000000, 6.000000) circle (0.214286);
			\draw[fill=red,line width = 0.03cm] (0.000000, 0.000000) circle (0.214286);
			\draw[fill=red,line width = 0.03cm] (0.000000, 3.000000) circle (0.214286);	
			\draw[fill=red,line width = 0.03cm] (3.000000, 3.000000) circle (0.214286);

%

			\draw[line width=0.25mm] (-3-0.214286,6) -- (-4.5,6);
			\draw[line width=0.25mm] (-3,6+0.214286) -- (-3,7.5);
			\draw[line width=0.25mm] (3+0.214286,3) -- (4.5,3);
			\draw[line width=0.25mm] (0,3+0.214286) -- (0,7.5);
			
			\draw[line width = 0.03cm] (-3.000000, 3.428571) circle (0.214286);
			\draw[line width = 0.03cm] (-3.000000, 3.857143) circle (0.214286);
			\draw[line width = 0.03cm] (-3.000000, 4.285714) circle (0.214286);
			\draw[line width = 0.03cm] (-3.000000, 4.714286) circle (0.214286);
			\draw[line width = 0.03cm] (-3.000000, 5.142857) circle (0.214286);
			\draw[line width = 0.03cm] (-3.000000, 5.571429) circle (0.214286);
			
			\draw[line width = 0.03cm] (-3.000000, 0.428571) circle (0.214286);
			\draw[line width = 0.03cm] (-3.000000, 0.734694) circle (0.214286);
			\draw[line width = 0.03cm] (-3.000000, 1.040816) circle (0.214286);
			\draw[line width = 0.03cm] (-3.000000, 1.346939) circle (0.214286);
			\draw[line width = 0.03cm] (-3.000000, 1.653061) circle (0.214286);
			\draw[line width = 0.03cm] (-3.000000, 1.959184) circle (0.214286);
			\draw[line width = 0.03cm] (-3.000000, 2.265306) circle (0.214286);
			\draw[line width = 0.03cm] (-3.000000, 2.571429) circle (0.214286);
			
			\draw[line width = 0.03cm] (-2.571429, 0.000000) circle (0.214286);
			\draw[line width = 0.03cm] (-2.265306, 0.000000) circle (0.214286);
			\draw[line width = 0.03cm] (-1.959184, 0.000000) circle (0.214286);
			\draw[line width = 0.03cm] (-1.653061, 0.000000) circle (0.214286);
			\draw[line width = 0.03cm] (-1.346939, 0.000000) circle (0.214286);
			\draw[line width = 0.03cm] (-1.040816, 0.000000) circle (0.214286);
			\draw[line width = 0.03cm] (-0.734694, 0.000000) circle (0.214286);
			\draw[line width = 0.03cm] (-0.428571, 0.000000) circle (0.214286);
			
			\draw[line width = 0.03cm] (0.000000, 0.428571) circle (0.214286);
			\draw[line width = 0.03cm] (0.000000, 0.857143) circle (0.214286);
			\draw[line width = 0.03cm] (0.000000, 1.285714) circle (0.214286);
			\draw[line width = 0.03cm] (0.000000, 1.714286) circle (0.214286);
			\draw[line width = 0.03cm] (0.000000, 2.142857) circle (0.214286);
			\draw[line width = 0.03cm] (0.000000, 2.571429) circle (0.214286);
			
			\draw[line width = 0.03cm] (0.428571, 0.000000) circle (0.214286);
			\draw[line width = 0.03cm] (0.696429, 0.000000) circle (0.214286);
			\draw[line width = 0.03cm] (0.964286, 0.000000) circle (0.214286);
			\draw[line width = 0.03cm] (1.232143, 0.000000) circle (0.214286);
			\draw[line width = 0.03cm] (1.500000, 0.000000) circle (0.214286);
			\draw[line width = 0.03cm] (1.767857, 0.000000) circle (0.214286);
			\draw[line width = 0.03cm] (2.035714, 0.000000) circle (0.214286);
			\draw[line width = 0.03cm] (2.303571, 0.000000) circle (0.214286);
			\draw[line width = 0.03cm] (2.571429, 0.000000) circle (0.214286);
			
			\draw[line width = 0.03cm] (3.000000, 0.428571) circle (0.214286);
			\draw[line width = 0.03cm] (3.000000, 0.734694) circle (0.214286);
			\draw[line width = 0.03cm] (3.000000, 1.040816) circle (0.214286);
			\draw[line width = 0.03cm] (3.000000, 1.346939) circle (0.214286);
			\draw[line width = 0.03cm] (3.000000, 1.653061) circle (0.214286);
			\draw[line width = 0.03cm] (3.000000, 1.959184) circle (0.214286);
			\draw[line width = 0.03cm] (3.000000, 2.265306) circle (0.214286);
			\draw[line width = 0.03cm] (3.000000, 2.571429) circle (0.214286);
			
			\draw[line width = 0.03cm] (0.428571,3) circle (0.214286);
			\draw[line width = 0.03cm] (0.857143,3) circle (0.214286);
			\draw[line width = 0.03cm] (1.285714,3) circle (0.214286);
			\draw[line width = 0.03cm] (1.714286,3) circle (0.214286);
			\draw[line width = 0.03cm] (2.142857,3) circle (0.214286);
			\draw[line width = 0.03cm] (2.571429,3) circle (0.214286);
			
			\draw[fill=blue,line width = 0.03cm] (-3.000000, 0.000000) circle (0.214286);				
			\draw[fill=blue,line width = 0.03cm] (-3.000000, 3.000000) circle (0.214286);
			\draw[fill=blue,line width = 0.03cm] (3.000000, 0.000000) circle (0.214286);

			\node[] (a1) at (-2.4,6) {$a_1$};
			\node[] (a2) at (-0.6,3) {$a_2$};
			\node[] (a3) at (3,3.4) {$a_3$};
			\node[] (a4) at (0,-0.5) {$a_4$};
			
			\node[] (b1) at (-3.5,3) {$b_1$};
			\node[] (b2) at (-3,-0.5) {$b_2$};
			\node[] (b3) at (3,-0.5) {$b_3$};
			
			\node[] (w1) at (-1.5,0.5) {$8$};
			\node[] (w2) at (-2.5,1.5) {$8$};	
			\node[] (w3) at (-2.5,4.5) {$6$};
			\node[] (w4) at (-0.5,1.5) {$6$};
			\node[] (w5) at (1.5,0.5) {$9$};
			\node[] (w6) at (2.5,1.5) {$8$};
			\node[] (w7) at (1.5,3.5) {$6$};

		\end{scope}

	\end{tikzpicture}
	\caption{Example of a construction of the unit disk representation of the graph $G'$ from the proof of \Cref{thm:is:npc:regular} for $r=3$. On the left, the original $3$-regular graph $G$ is embedded into a grid (the half-edges represent the rest of the graph). On the right, the subdivision was made. The red disks correspond to the original vertices of the graph and blue disks correspond to the internal grid points contained in the polygonal chains representing the edges. These are the disks $D_2,\ldots,D_{g-1}$ for the corresponding edges. Consider the edge $e=\{a_1,a_4\}$. The red disks at $a_1$ and $a_4$ are the disks $D_1$ and $D_4$, respectively, and the blue disks at $b_1$ and $b_2$ are the disks $D_2$ and $D_3$, respectively. The empty disks correspond to the disks $E_j$ from \Cref{lem:is:npc:construction_technical_lemma1}.
	The numbers next to the edges correspond to the values~$w_i$ from \Cref{lem:is:npc:construction_technical_lemma2} (and are equal to the number of empty disks $E_j$ between (filled) red and blue disks). For the edge~$e$, we have $g=4$ grid points contained in the polygonal chain $\mathcal{E}(e)$, thus we are in the case $g=4 \mod 6$ from \Cref{lem:is:npc:construction_technical_lemma2} thus $w_1=w_2=8$ and $w_3=6$ for this particular edge. The total number of disks on the subdivided edge $e$ is thus $2+8+8+6=24=0\mod 6$ and thus $q_e=4$.
	}\label{fig:udg_construction}
\end{figure}
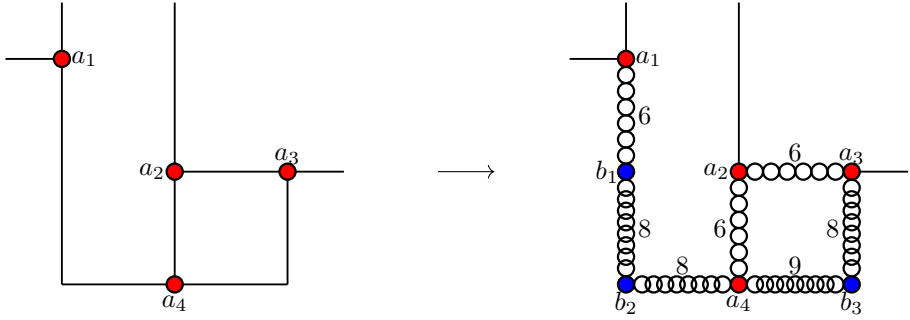

This completes the description of the unit disk representation for $G'$. An example of the construction is given in \Cref{fig:udg_construction}. Note that \Cref{lem:is:npc:construction_technical_lemma1} ensures that the disks are placed in between the disks $D_i,D_{i+1}$ starting from $s_1=(d,0)$ and ending at $s_\ell=(1-d,0)$. This implies that for any edge $e=\{u,v\}$ the disks adjacent to $D_1$ and $D_g$ will not intersect any other disks representing other subdivided edges (in particular those with endpoints $u$ or $v$). 

We are now ready to finally prove \Cref{thm:is:npc:regular}.

\begin{proof}[of \Cref{thm:is:npc:regular}]
	Reduce from \IS on $r$-regular planar graphs. This setting is \NPh~\citep{FleischnerSS2010}. Let $(G,k)$ be an instance of \IS where $G$ is $r$-regular planar graph. Construct an instance $(G',k')$ of \IS where $G'$ is $r$-regular unit disk graph as described above. Combining \Cref{lem:is:npc:onedir,lem:is:npc:twodir} the instances $(G,k)$ and $(G',k')$ are equivalent. It remains to argue that the reduction is polynomial. Computation of the rectilinear embedding $\mathcal{E}$ for $G$ can be done in polynomial time by \Cref{thm:valiant}. Computation of the numbers $w_i$, $g$ and $q_e$ can also be done in polynomial time. What is left to show is that the numbers $q_e$ are also polynomially bounded by the size of~$G$. By \Cref{thm:valiant} the area of $\mathcal{E}$ satisfies $\operatorname{Area}(\mathcal{E})\leq \Oh{|V(G)|^2}$. For any edge~$e$, the number of grid points contained in the polygonal chain induced by $\mathcal{E}(e)$ is at most $\operatorname{Area}(\mathcal{E})$. It follows that $g\leq \Oh{|V(G)|^2}$ for any edge $e$. By construction we have $w_i\leq 9$ for any $g\geq 2$ and $i\in[g-1]$. Thus, we have
	$$
	q_e\leq \frac{1}{6}\left(g-2+9(g-1)\right)\leq\frac{1}{6}(10g-11)\leq \frac{10}{6}g\leq\Oh{|V(G)|^2}.
	$$
	The number of newly added vertices is at most $\sum_{e\in E}10q_e$ and since $q_e$ is polynomial in $|V(G)|$, the reduction is indeed polynomial.
\end{proof}

Using \Cref{thm:is:npc:regular}, we can easily show that \TSS remains \NPc even if the threshold function is bounded by a constant. We show a construction for the case where all thresholds are exactly $3$. It is not hard to see that our result holds for every constant $c\geq 3$. Whenever thresholds are bounded by a constant $c'$, then they are certainly bounded by any constant $c\geq c'$.

\begin{corollary}
	\TSS is \NPc even if the underlying graph is a unit disk graph and all thresholds are bounded by a constant $c\geq 3$.
\end{corollary}
\begin{proof}
	By \Cref{thm:is:npc:regular}, \IS is \NPc when restricted to the class of $3$-regular unit disk graphs. The same \NPhness holds for \VC. We reduce from \VCshort restricted to such instances. Let $(G,k)$ be an instance of \VC and $G$ be a $3$-regular graph. Set $G'=G,k'=k$ and $\thr(v)=3$ for each $v\in V(G)$. Since $G'$ is $3$-regular this is the case of unanimous thresholds. As noted in \Cref{lem:tss:unanimous:vc:eq}, instances of \TSS with unanimous thresholds are equivalent to the \VC problem, so the theorem follows.
\end{proof}

The above result can be generalized for infinitely many constants $r$. Given an instance $(G,k)$ of \IS where the underlying graph is $r$-regular, replacing each vertex by a clique~$K_q$ makes the graph $(q(r+1)-1)$-regular. If we denote the new graph by $G_q$ it is not hard to see that thes instances $(G_q,k)$ and $(G,k)$ of \IS are equivalent. Replacing vertices by cliques can be easily achieved in intersection graph classes, in particular, unit disk graphs.

\begin{corollary}\label{cor:inf}
	If \IS is \NPh on the class of $r$-regular graphs, then \IS is \NPh in the class of $(q(r+1)-1)$-regular graphs for any positive integer $q$.
\end{corollary}

Combining \Cref{cor:inf} with \Cref{thm:is:npc:regular} we obtain the following.

\begin{corollary}
	\IS is \NPh even if the underlying graph is an $r$-regular unit disk graph where $r$ is positive integer and $r=-1\mod 4$ or $r=-1\mod 5$.
\end{corollary}

\begin{remark}\label{remark:is}
	We remark that this approach does not prove \NPhness of \IS for \emph{all} constants $r\geq 3$ (note that for $r\leq 2$ the problem is in \P). The first value of $r$ unknown to us is $r=5$. Note that \IS is \NPh on planar $5$-regular graphs~\citep{Amiri2021}, however \Cref{thm:valiant} is not applicable since $\Delta G = 5$ in this case.
	
	By using \Cref{cor:inf} and explicit proof for $r=3,4$ we obtained \NPhness for infinitely many constants $r\geq 3$. Unfortunately, this approach can never be used to show \NPhness for \emph{all} constants~$r\geq 3$. To see this, note that even if we explicitly show \NPhness for any number of constants $r_1,\ldots,r_k$, we can pick a large enough prime number $p$ satisfying $p>r_i+1$ for all $i\in[k]$. Observe that explicitly proving \NPhness for $r_i$-regular graphs implies \NPhness for $r$-regular graphs with $r=-1\mod r_i+1$. Now, the \NPhness for $r=p-1$ is not implied by \NPhness for $r_1,\ldots,r_k$ together with \Cref{cor:inf} since otherwise $p-1=-1\mod r_j+1$ for some $j\in[k]$ which implies that $r_j+1$ divides $p$, contradicting the choice of $p$.
\end{remark}

\subsection{Thresholds bounded by 2 in planar graphs}

Following historical development in the study of \TSSshort, it remains to show the complexity of the problem if all thresholds are bounded by $2$ and the underlying graph is a~unit disk graph.

We first establish the \NPhness for planar graphs with $\Delta G \leq 3$ and then utilize this reduction to show \NPhness for the class of grid graphs. The \NPhness for unit disk graphs follows.

Note that the restriction on maximum degree or the threshold function cannot be further strengthened. For $t(v)\leq 1$ or $\Delta G \leq 2$ the problem is trivial, for $t(v)= 2$ and $\Delta G \leq 3$ the problem is polynomial-time solvable in general graphs~\citep{KynclLV2017}.

\begin{theorem}\label{thm:tss:planar:npc:const}
	\TSS is \NPc even when the underlying graph is planar with maximum degree $\Delta G \leq 3$ and all thresholds are at most $2$.
\end{theorem}

\subsection*{Construction}

The reduction is from the \rptSAT. Let $\varphi$ be an instance of \rptSAT and let $\operatorname{Var}(\varphi)$ denote the set of variables of $\varphi$ and $\mathcal{C}(\varphi)$ the set of clauses of $\varphi$. Recall that each clause consists of at most $3$ literals and each variable occurs exactly three times in $\varphi$ -- twice as positive literal and once as negative literal. The reduction consists of two types of gadgets:

\begin{description}
	\item[Variable gadget] Given a~variable~$x\in\operatorname{Var}(\varphi)$, the \emph{variable gadget} for $x$ is the planar graph depicted in \Cref{fig:variable_gadget}. We refer to this graph as $\operatorname{VG}(x)$. For simplicity of notation we write $V(x)$ instead of $V(\operatorname{VG}(x))$ for the set of vertices of $\operatorname{VG}(x)$. The notable vertices of the gadget are $T_x$, $F_x$, $t_x^1$, $t_x^2$, and $f_x$. The idea is that the vertices $T_x$ and $F_x$ stand for the truth assignment of variable $x$ while the vertices $t_x^1$, $t_x^2$ and $f_x$ represent the two positive and one negative occurrence of the variable $x$. These vertices serve to connect the variable gadget with the respective clause gadgets.
	\item[Clause gadget] Given a clause $C\in\mathcal{C}(\varphi)$, the clause gadget for $C$ consists of single vertex $w_C$ which is connected to the corresponding literal vertices that are contained in $C$. We refer to this gadget as $\operatorname{CG}(C)$.
\end{description}

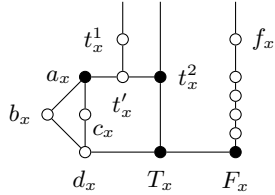
\begin{figure}[ht]
	\centering
	\begin{tikzpicture}[every node/.style={draw,circle,inner sep=1.5pt}]
		\tikzstyle{f2} = [fill=black]
		\node[f2,label=0:{\small$t_x^2$}] (ti2) at (0,0){};
		\node[label=270:{\small$t_x'$}] (tip) at (-0.5,0){};
		\node[label=180:{\small$t_x^1$}] (ti1) at (-0.5,0.5){};
		\node[f2,label=270:{\small$T_x$}] (Ti) at (0,-1){};
		\node[label=0:{\small$f_x$}] (fi) at (1,0.5){};
		\node[]    (fi1) at (1,0) {};
		\node[]    (fi2) at (1,-0.25) {};
		\node[]    (fi3) at (1,-0.5) {};
		\node[]    (fi4) at (1,-0.75) {};
		\node[f2,label=270:{\small$F_x$}] (Fi) at (1,-1){};
		
		\node[f2,label=180:{\small$a_x$}] (ai) at (-1,0){};
		\node[label=180:{\small$b_x$}]     (bi) at (-1.5,-0.5){};
		\node[label=330:{\small$c_x$}]     (ci) at (-1,-0.5){};
		\node[label=270:{\small$d_x$}]     (di) at (-1,-1){};
		\draw (ai) -- (tip);
		\draw (tip) -- (ti2);
		\draw (ai) -- (bi);
		\draw (ai) -- (ci);
		\draw (ci) -- (di);
		\draw (bi) -- (di);
		\draw (di) -- (Ti);
		\draw (Ti) -- (Fi);
		\draw (Fi) -- (fi4);
		\draw (fi4) -- (fi3);
		\draw (fi3) -- (fi2);
		\draw (fi2) -- (fi1);
		\draw (fi1) -- (fi);
		\draw (ti2) -- (Ti);
		\draw (tip) -- (ti1);
		\draw (ti1) -- (-0.5,1);
		\draw (ti2) -- (0,1);
		\draw (fi) -- (1,1);
	\end{tikzpicture}
	\caption{Schematic representation of the variable gadget $\operatorname{VG}(x)$ for variable $x$. The filled vertices have threshold $2$, while the white vertices have threshold $1$. Note also that the half-edges illustrate the fact that the gadget is connected with the rest of the graph only via $t_x^1,t_x^2$ and $f_x$.}
	\label{fig:variable_gadget}
\end{figure}

We are now ready to construct an instance $(G,t,k)$ of \TSS. Start with the incidence graph $G_\varphi$. For every variable $x$ we replace the vertex $v_x$ by the variable gadget $\operatorname{VG}(x)$ and we identify each clause vertex $v_{C}$ with the vertex $w_C$, i.e., with the gadget $\operatorname{CG}(C)$. Next, for all variables $x$, we connect all literal vertices of $\operatorname{VG}(x)$ with the corresponding clause gadgets. More precisely, if $C_1,C_2$ are the clauses where $x$ occurs as positive literal and $C_3$ is the clause where $x$ occurs as negative literal, we add edges $\{t_x^1,w_{C_1}\}$, $\{t_x^2,w_{C_2}\}$ and $\{f_x,w_{C_3}\}$.

It remains to set the thresholds and~$k$. In the variable gadgets the filled vertices have threshold equal to~$2$, while the white vertices have threshold equal to~$1$. In the clause gadgets we set $t(w_C)=1$ for all clauses $C$. Finally, we set $k = n$.

Observe that $G$ is a planar graph. To see this, start with a planar drawing of $G_\varphi$ and replace the vertices of $G_\varphi$ with the corresponding gadgets. The only problem could be with the edges coming from the vertices $t_x^1$, $t_x^2$ and $f_x$. However, for a variable $x$ that occurs in clauses $C_1$, $C_2$ and $C_3$, no matter what the radial order\footnote{More precisely, the order of the points corresponding to the vertices $v_{C_1}$, $v_{C_2}$, $v_{C_3}\in V(G_\varphi)$ given by radially sorting them around the point corresponding to the vertex $v_{x}\in V(G_\varphi)$ in the original drawing.} of the vertices $w_{C_1}$, $w_{C_2}$, $w_{C_3}$ is (with respect to the planar drawing of $G_\varphi$), it is always possible to draw the edges from $t_x^1$, $t_x^2$ and $f_x$ to the corresponding clause gadgets in such a way that we do not create any crossings. We can always, without loss of generality, swap $w_{C_1}$ with $w_{C_2}$, i.e. instead of having the edges $\{t_x^1,w_{C_1}\}$, $\{t_x^2,w_{C_2}\}$ we can have the pair of edges $\{t_x^1,w_{C_2}\}$, $\{t_x^2,w_{C_1}\}$ and if needed the edge coming from $t_x^1$ can encircle the entire gadget in the drawing and leave the gadget to the right of the edge coming from $f_x$.

Moreover, we have $\Delta G \leq 3$, and the thresholds are at most $2$, as promised.

Before showing the equivalence of the instances $(G,t,k)$ and $\varphi$, we establish some basic properties of the variable gadget. Properties of the clause gadget are clear, since it is a single vertex with threshold $1$.

\begin{lemma}\label{lem:gadget_properties}
	Let $x\in \operatorname{Var}(\varphi)$ be a variable. The gadget $\operatorname{VG}(x)$ has the following properties:
	\begin{enumerate}[label=\roman*)]
		\item \label{gadget_1} If the vertex $T_x$ is active, then after $5$ rounds, the vertices $a_x,b_x,c_x,d_x,t_x',t_x^1,t_x^2$ are necessarily active.
		\item \label{gadget_2} If the vertex $F_x$ is active, then after $5$ rounds, all vertices on the $f_x$-$F_x$-path are necessarily active.
		\item \label{gadget_3} The vertices $T_x$ and $F_x$ never become active unless some vertex $v\in  V(x)$ is active.
		
		\item \label{gadget_4} If the vertex $F_x$ is active and the neighbors of $t_x^1$ and $t_x^2$ outside $\operatorname{VG}(x)$ become active, then all vertices in $\operatorname{VG}(x)$ are eventually active.
		
		\item \label{gadget_5} If the vertex $T_x$ is active and the neighbor of $f_x$ outside $\operatorname{VG}(x)$ becomes active, then all vertices in $\operatorname{VG}(x)$ are eventually active.
	\end{enumerate}
\end{lemma}
\begin{proof}
	The properties \ref{gadget_1} and \ref{gadget_2} clear from the construction of the gadget. To prove \ref{gadget_3}, suppose that no vertices inside $\operatorname{VG}(x)$ are active and consider the extreme case that all the vertices in $N(V(x))\setminus V(x)$ are active. The vertices $t_x^1,t_x',t_x^2$ become active. However, both $a_x$ and $T_x$ have threshold $2$ and both of them have only one active neighbor since $d_x,b_x,c_x$ are not active by assumption. Hence, $T_x$ never becomes active. Similarly, if the neighbor of $f_x$ outside $V$ is active, it only activates the neighbor of $F_x$ with threshold one but $F_x$ has threshold $2$ so he doesn't get active as $T_x$ is not active either. Thus $F_x$ will not become active.
	
	We now turn our attention to property \ref{gadget_4}. Suppose that $F_x$ is active and the neighbors of $t_x^1$ and $t_x^2$ outside $V$ become active in round $r$. Then, in round $r+1$ the vertex $t_x^1$ becomes active and in round $r+2$ the vertex $t_x'$ becomes active. In round $r+3$ the vertex $t_x^2$ becomes active and since $T_x$ now has two active neighbors it becomes active. Finally, the vertices $d_x,b_x,c_x$ and $a_x$ become active. Note that the vertices on $F_x$-$f_x$-path are also active since $F_x$ is active. Thus, all vertices of $V(x)$ are active.
	
	It remains to prove property \ref{gadget_5}. Suppose that $T_x$ is active and $f_x$ becomes active in round $r$. At latest in round $r+5$ the vertex $F_x$ becomes active. By property \ref{gadget_1} the remaining vertices of the gadget become active.
\end{proof}

We now establish the equivalence between the formula $\varphi$ and the constructed instance $(G,t,k)$.

\begin{claim}\label{claim:equivalencedir1}
	If $\varphi$ is satisfiable, then $(G,t,k)$ is a \emph{yes}-instance of \TSSshort.
\end{claim}
\begin{proof}
	Let $\varphi$ be satisfiable and let $\pi$ be a satisfying assignment for $\varphi$. We create a target set $S$ as follows. For each variable $x$ we add either $T_x$ if $\pi(x)=1$ or $F_x$ if $\pi(x) = 0$. Observe that $|S|=n=k$. It remains to show that $S$ is a target set.
	
	To see this, observe that for any variable $x$ by the properties \ref{gadget_1} and \ref{gadget_2} from \Cref{lem:gadget_properties} $T_x$ activates $t_x^1$ and $t_x^2$ and $F_x$ activates $f_x$ in $5$ rounds. In the sixth round, all clauses become active. Indeed, because $\pi$ is a satisfying assignment, all of them become active. Finally, by applying the properties \ref{gadget_4} and \ref{gadget_5} from \Cref{lem:gadget_properties}, since all neighbors of $t_x^1$, $t_x^2$ and $f_x$ outside $V(x)$ are active for all $i$, all remaining vertices in all variable gadgets are eventually active. Thus, $S$ is a target set and $(G,t,k)$ is a \emph{yes}-instance.
\end{proof}

\begin{claim}\label{claim:equivalencedir2}
	If $(G,t,k)$ is a \emph{yes}-instance of \TSSshort, then $\varphi$ is satisfiable.
\end{claim}

Before proving \Cref{claim:equivalencedir2} we first make several observations about the structure of the solution to the instance $(G,t,k)$. In what follows, we denote by $S$ an arbitrary target set in $G$ of size at most $k$. Recall that $k=n$.

\begin{observation}\label{obs:vertex_in_vgadget}
	For every variable $x$ we have $S\cap V(x)\neq \emptyset$.
\end{observation}
\begin{proof}
	Suppose otherwise, and let $x$ be a variable such that $S\cap V(x)=\emptyset$. Note that by property \ref{gadget_3} from \Cref{lem:gadget_properties}, the vertices $T_x$ and $F_x$ never become active even if all the vertices outside $V(x)$ are active. This contradicts the assumption that $S$ is a target set.
\end{proof}

\begin{observation}\label{obs:size_of_S_in_gadget}
	For all variables $x$ we have $|S\cap V(x)| = 1$.
\end{observation}
\begin{proof}
	 By \Cref{obs:vertex_in_vgadget}, $S$ must contain at least one vertex from each variable gadget. On the other hand, if there is a variable gadget containing at least $2$ vertices from $S$ then, by the pigeonhole principle, there is a variable gadget containing no vertices from $S$, contradicting \Cref{obs:vertex_in_vgadget} as there are exactly~$n$ variable gadgets.
\end{proof}

\begin{observation}\label{obs:the_vertex_in_gadget}
	There exists a target set $S'$ such that for all variables $x$ we have $S'\cap \{T_x,F_x\}\neq \emptyset$ and $|S'| = |S|$.
\end{observation}
\begin{proof}
	Process the variable gadgets independently one by one. Let $\operatorname{Var}(\varphi)=\{x_1,\ldots,x_n\}$. Formally\footnote{We use superscripts to avoid confusion with the activation process.}, we build target sets $S^i$ for $i=0,1,\ldots,n$ and set $S'=S^n$. Start with $S^0 = S$. Let $i\geq 1$ and consider the variable gadget $\operatorname{VG}(x_i)$. If $S^{i-1}\cap \{T_{x_i},F_{x_i}\}\neq \emptyset$, there is nothing to do, i.e., we set $S^i = S^{i-1}$. Otherwise, observe that $S\cap V(x_i)=S_{i-1}\cap V(x_i)$. By \Cref{obs:size_of_S_in_gadget} we a unique vertex $u_{x_i}\in S\cap V(x_i)$. We distinguish two cases:
	\begin{description}
		\item[Case 1] The vertex $u_{x_i}$ lies on the $f_{x_i}$-$F_{x_i}$-path in $\operatorname{VG}(x_i)$. Note that we can replace $u_{x_i}$ by $F_{x_i}$ and this does not change the fact that $u_{x_i}$ eventually becomes active by property \ref{gadget_2} from \Cref{lem:gadget_properties}. In this case we set $S^i = S^{i-1}\setminus \{u_{x_i}\} \cup \{F_{x_i}\}$.
		\item[Case 2] The vertex $u_i$ is one of $a_{x_i},b_{x_i},c_{x_i},d_{x_i},t_{x_i}',t_{x_i}^1,t_{x_i}^2$. Note that we can replace $u_{x_i}$ by $T_{x_i}$ and this does not change the fact that $u_{x_i}$ eventually becomes active by property \ref{gadget_1} from \Cref{lem:gadget_properties}. In this case we set $S^i = S^{i-1}\setminus \{u_{x_i}\} \cup \{T_{x_i}\}$.
	\end{description}
	This finishes the proof of \Cref{obs:the_vertex_in_gadget}.
\end{proof}

We are now ready to give proof of \Cref{claim:equivalencedir2}.

\begin{proof}[of \Cref{claim:equivalencedir2}]
	Let $(G,t,k)$ be a \emph{yes}-instance and let $S$ be a target set in $G$ of size at most $k$. By \Cref{obs:the_vertex_in_gadget} we can assume, without loss of generality, that for all variables $x$ the unique vertex $u_x$ in $S\cap V(x)$ is either $T_x$ or $F_x$. We now construct a satisfying assignment $\pi$ for $\varphi$ in the obvious way. We set $\pi(x) = 0$ if $u_x=F_x$ and $\pi(x)= 1$ if $u_x=T_x$. The only thing left is to show that $\pi$ is indeed a satisfying assignment for $\varphi$.
	
	For the sake of contradiction, suppose that $\pi$ is not a satisfying assignment. Thus, there is a clause $C$ not satisfied by $\pi$. Note that $w_C\notin S$ because otherwise there is a variable gadget $\operatorname{VG}(x)$ with $S\cap V(x) = \emptyset$ which contradicts \Cref{obs:vertex_in_vgadget}. Since $S$ is a target set and $w_C \notin S$ and $t(w_C)=1$ there must be a round $r$ in which one of the neighbors of $w_C$ becomes active. We show that this is impossible.
	
	We handle the positive and negative literals inside $C$ separately. 
	
	Let $x$ be a variable occurring as positive literal in $C$ and let $y$ be a variable occurring as a negative literal in $C$. By assumption we have $\pi(x)=0$ and $\pi(y) = 1$.

	Consider the connection of $w_C$ to the variable gadget $\operatorname{VG}(x)$. Since $x$ occurs as positive literal in $C$, the vertex $w_C$ is connected to $t_x^1$ or $t_x^2$. By construction of $\pi$ we have $S\cap V(x)=\{F_x\}$. The only way that $t_x^1$ or $t_x^2$ could be activated is via $T_x$. However $T_x$ has only one active neighbor and since no other vertices from $\operatorname{VG}(x)$ are in $S$, $t_x^1$ nor $t_x^2$ get activated. Note that even if $w_C$ is connected to $t_x^2$ and the vertex $t_x^1$ gets activated via the other clause connected to $t_x^1$, this doesn't activate the vertex $t_x^2$. On the other hand if $w_C$ was connected with $\operatorname{VG}(x)$ via $t_x^1$, activation of the clause connected via $t_x^2$ won't even activate the vertex $t_x^2$, thus $t_x^1$ is neither active.
	
	Finally, consider the connection of $w_C$ to the variable gadget $\operatorname{VG}(y)$. The vertex $w_C$ is connected to~$f_x$ and $S\cap V(y)=\{T_y\}$. The only way that $f_x$ is activated is via $F_x$. However $F_x$ has only one active neighbor and by the same argument as above, $f_x$ doesn't get activated.
	
	To conclude, $w_C$ does not get activated and this contradicts the fact that $S$ is a target set, thus $\pi$ is satisfying assignment for $\varphi$.
\end{proof}

We can finally give proof of  \Cref{thm:tss:planar:npc:const}.

\begin{proof}[of \Cref{thm:tss:planar:npc:const}]
	Reduce from \rptSAT which is \NPh~\citep{DahlhausJPSY1994}. Let $\varphi$ be an instance of \rptSAT consisting of $n$ variables and $m$ clauses. Construct an instance $(G,t,k)$ from $\varphi$ as described above. Combining \Cref{claim:equivalencedir1,claim:equivalencedir2}, the instances $\varphi$ and $(G,t,k)$ are equivalent. To conclude, it remains to say that the resulting graph $G$ has exactly $m+14n$ vertices, thus the reduction is indeed polynomial.
\end{proof}

\subsection{Thresholds bounded by 2 in Grid Graphs and Unit Disk Graphs}

In the previous section we showed \NPhness of \TSS when the underlying graph is planar and has maximum degree at most~$3$ and the thresholds are at most $2$. We now utilize this result to show \NPhness in the same setting for the class of grid graphs. Note that hardness for unit disk graphs follows since by \Cref{lem:grids_are_udg} every grid graph is a unit disk graph. Let us begin with a few observations about how edge subdivisions affect target sets.

\begin{observation}\label{obs:threshold_one_vertices}
	Let $G=(V,E)$ be a~graph and $t\colon V \rightarrow \mathbb{N}$ a~threshold function, $S\subseteq V$ a~target set, and let $v\in S$ be a~vertex with $t(v)\leq 1$ and $\deg (v) \geq 1$. Then for any $u\in N(v)$, the set $S\setminus \{v\}\cup \{u\}$ is also a~target set.
\end{observation}

\begin{observation}\label{obs:subdivided_edges}
	Let $G=(V,E)$ be a~graph and $t\colon V \rightarrow \mathbb{N}$ a~threshold function, and let~$e\in E$. Let~$G'$ be a~graph that results from~$G$ by subdividing the edge~$e$ once, creating a new vertex $v'\notin V$. Let~$t'\colon V(G')\rightarrow \mathbb{N}$ be defined by $t'(v')=1$ and $t'(v)=t(v)$ for $v\neq v'$. Then the following holds:
	\begin{enumerate}[label=\roman*)]
		\item If $S$ is a~target set for $G$ with respect to $t$, then $S$ is also a~target set for $G'$ with respect to $t'$.\label{subdiv1}
		\item If $S'$ is a~target set for $G'$ with respect to $t'$, then there exists a target set $S$ for $G$ with respect to $t$ and $|S| = |S'|$.\label{subdiv2}
	\end{enumerate}
\end{observation}

\begin{theorem}\label{thm:tss_npc_constant_grid}
	\TSS is \NPc even if the underlying graph is a~grid graph with maximum degree at most $3$ and all thresholds are at most~$2$.
\end{theorem}
\begin{proof}
	We reduce from \TSS on planar graphs with maximum degree $3$ and thresholds at most~$2$. This setting is \NPh by
	\Cref{thm:tss:planar:npc:const}. Let $(G,t,k)$ be an instance of \TSSshort where $G$ is planar and $\Delta G \leq 3$. By \Cref{thm:valiant} there is a rectilinear embedding of $G$ of polynomial area and computable in polynomial time. Fix one such embedding and denote it by $\mathcal{E}$. We now modifiy the graph~$G$ as follows. For an edge $e\in E(G)$, let $\mathcal{E}(e)=(p_1,\ldots,p_g)$. We subdivide the edge $e$ exactly $g-2$ times (see~\Cref{fig:subdividing_edges_to_subgraph}). Note that the case $g=2$ vacuously corresponds to no subdivision. After this step, the graph is (not necessarily induced) subgraph of a grid. To make it induced, we further simultaneously subdivide all edges exactly once (see \Cref{fig:subdividing_edges_to_induced}). After this step, the resulting graph is indeed an induced subgraph of a~grid (i.e., a~grid graph). We set the thresholds of all newly created vertices to $1$. Let $G'$ denote the resulting graph, $t'\colon V(G')\rightarrow \mathbb{N}$ the new threshold function and set $k'=k$.
	\begin{claim}
		$(G,t,k)$ is a~\emph{yes}-instance of \TSS if and only if $(G',t',k')$ is a~\emph{yes}-instance of \TSS.
	\end{claim}
	\begin{proof}
		Let $(G,t,k)$ be a~\emph{yes}-instance and let $S\subseteq V(G)$ be a target set of size at most $k$. Inductively, for each subdivision, apply \ref{subdiv1} from \Cref{obs:subdivided_edges}. It follows that $S$ is also a target set with respect to~$t'$ and is of size at most $k=k'$, thus $(G',t',k')$ is a~\emph{yes}-instance.
		
		On the other hand, let $(G',t',k')$ be a \emph{yes}-instance and let $S'\subseteq V(G')$ be a target set of size at most~$k'$. Inductively for each subdivision, apply \ref{subdiv2} from \Cref{obs:subdivided_edges}. Observe that in each step we get a target set $S$ with the same size. It follows that there is a target set~$S$ with respect to~$t$ of size $k'=k$, thus $(G,t,k)$ is a~\emph{yes}-instance.
	\end{proof}
	
	To finish the proof, we notice that the rectilinear embedding can be computed in polynomial time by \Cref{thm:valiant} and its area is at most $\Oh{|V|^2}$. It follows that in both steps of the construction, we only added at most $\Oh{|V|^2}$ many new vertices, and thus the size of $G'$ is at most polynomial in the size of $G$. This implies that the reduction is polynomial. The theorem follows.
	
\end{proof}

\begin{figure}[ht]
	\centering
	\begin{tikzpicture}[every node/.style={draw,circle,inner sep=1.5pt}, f2/.style={fill=black}]
	
	\node[f2] (w1) at (-5,2){};
	\node[f2] (w2) at (-4,1){};
	\node[f2] (w3) at (-3,1){};
	\node[f2] (w4) at (-5,0){};
	\node[f2] (w6) at (-3,0){};
	\node[f2] (w7) at (-5,-1){};
	\node[f2] (w8) at (-6,-2){};
	\node[f2] (w9) at (-6,1){};
	\draw (w4) -- (-5,1);
	\draw (-5,1) -- (w2);
	\draw (w2) -- (w3);
	\draw (w2) -- (-4,2);
	\draw (-4,2) -- (w1);
	\draw (w9) -- (-6,0);
	\draw (-6,0) -- (w4);
	\draw (w4) -- (w7);
	\draw (w7) -- (-3,-1);
	\draw (-3,-1) -- (w6);
	\draw (-6,-1) -- (w7);
	\draw (-6,-1) -- (w8);

	\draw[->] (-2.25,0) -- (-1.75,0);

	\node[f2] (v1) at (0,2){};
	\node[f2] (v2) at (1,1){};
	\node[f2] (v3) at (2,1){};
	\node[f2] (v4) at (0,0){};
	\node[f2] (v6) at (2,0){};
	\node[f2] (v7) at (0,-1){};
	\node[f2] (v8) at (-1,-2){};
	\node[f2] (v9) at (-1,1){};
	
	\node[] (n1) at (-1,0){};
	\node[] (n2) at (-1,-1){};
	\node[] (n3) at (2,-1){};
	\node[] (n4) at (0,1){};
	\node[] (n5) at (1,2){};
	\node[] (n6) at (1,-1){};

	\draw (v4) -- (n4);
	\draw (n4) -- (v2);
	\draw (v2) -- (v3);
	\draw (v2) -- (n5);
	\draw (n5) -- (v1);
	\draw (v9) -- (n1);
	\draw (n1) -- (v4);
	\draw (v4) -- (v7);
	\draw (v7) -- (n6);
	\draw (n6) -- (n3);
	\draw (n3) -- (v6);
	\draw (n2) -- (v7);
	\draw (n2) -- (v8);

	\end{tikzpicture}
	\caption{Transformation of a~planar graph with maximum degree $3$ into a~subgraph of a~grid by subdividing edges at the internal points of the polygonal chains. Filled vertices correspond to the vertices of the original graph and the white ones are the newly created vertices.}
	\label{fig:subdividing_edges_to_subgraph}
\end{figure}
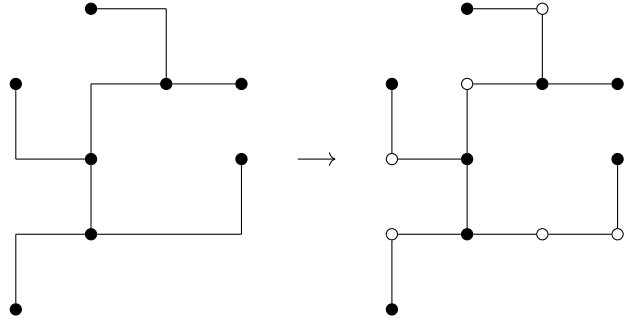

\begin{figure}[ht]
	\centering
	\begin{tikzpicture}[every node/.style={draw,circle,inner sep=1.5pt}, f2/.style={fill=black}]
	
	\node[f2] (v1) at (-5,2){};
	\node[f2] (v2) at (-4,1){};
	\node[f2] (v3) at (-3,1){};
	\node[f2] (v4) at (-5,0){};
	\node[f2] (v6) at (-3,0){};
	\node[f2] (v7) at (-5,-1){};
	\node[f2] (v8) at (-6,-2){};
	\node[f2] (v9) at (-6,1){};
	
	\node[f2] (n1) at (-6,0){};
	\node[f2] (n2) at (-6,-1){};
	\node[f2] (n3) at (-3,-1){};
	\node[f2] (n4) at (-5,1){};
	\node[f2] (n5) at (-4,2){};
	\node[f2] (n6) at (-4,-1){};
	
	\draw (v4) -- (n4);
	\draw (n4) -- (v2);
	\draw (v2) -- (v3);
	\draw (v2) -- (n5);
	\draw (n5) -- (v1);
	\draw (v9) -- (n1);
	\draw (n1) -- (v4);
	\draw (v4) -- (v7);
	\draw (v7) -- (n6);
	\draw (n6) -- (n3);
	\draw (n3) -- (v6);
	\draw (n2) -- (v7);
	\draw (n2) -- (v8);

	\draw[->] (-2.25,0) -- (-1.75,0);

	
	\node[f2] (v1) at (0,2){};
	\node[f2] (v2) at (1,1){};
	\node[f2] (v3) at (2,1){};
	\node[f2] (v4) at (0,0){};
	\node[f2] (v6) at (2,0){};
	\node[f2] (v7) at (0,-1){};
	\node[f2] (v8) at (-1,-2){};
	\node[f2] (v9) at (-1,1){};
	
	\node[f2] (n1) at (-1,0){};
	\node[f2] (n2) at (-1,-1){};
	\node[f2] (n3) at (2,-1){};
	\node[f2] (n4) at (0,1){};
	\node[f2] (n5) at (1,2){};
	\node[f2] (n6) at (1,-1){};

	\node[] (z0) at (0.0,0.5){};
	\draw (v4) -- (z0);
	\draw (z0) -- (n4);
	\node[] (z1) at (0.5,1.0){};
	\draw (n4) -- (z1);
	\draw (z1) -- (v2);
	\node[] (z2) at (1.5,1.0){};
	\draw (v2) -- (z2);
	\draw (z2) -- (v3);
	\node[] (z3) at (1.0,1.5){};
	\draw (v2) -- (z3);
	\draw (z3) -- (n5);
	\node[] (z4) at (0.5,2.0){};
	\draw (n5) -- (z4);
	\draw (z4) -- (v1);
	\node[] (z5) at (-1.0,0.5){};
	\draw (v9) -- (z5);
	\draw (z5) -- (n1);
	\node[] (z6) at (-0.5,0.0){};
	\draw (n1) -- (z6);
	\draw (z6) -- (v4);
	\node[] (z8) at (0.0,-0.5){};
	\draw (v4) -- (z8);
	\draw (z8) -- (v7);
	\node[] (z9) at (0.5,-1.0){};
	\draw (v7) -- (z9);
	\draw (z9) -- (n6);
	\node[] (z10) at (1.5,-1.0){};
	\draw (n6) -- (z10);
	\draw (z10) -- (n3);
	\node[] (z11) at (2.0,-0.5){};
	\draw (n3) -- (z11);
	\draw (z11) -- (v6);
	\node[] (z12) at (-0.5,-1.0){};
	\draw (n2) -- (z12);
	\draw (z12) -- (v7);
	\node[] (z13) at (-1.0,-1.5){};
	\draw (n2) -- (z13);
	\draw (z13) -- (v8);


	\end{tikzpicture}
	\caption{Transformation of a graph, that is (not necessarily induced) subgraph of a~grid into a~graph that is induced subgraph of a~grid (i.e., a~grid graph) by subdividing all edges exactly once. Filled vertices correspond to the vertices of the original graph and the white ones are the newly created vertices.}
	\label{fig:subdividing_edges_to_induced}
\end{figure}
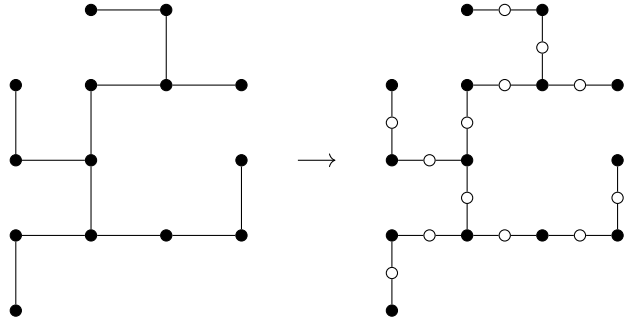

As the class of grid graphs is a~subclass of the unit disk graphs, we obtain \NPhness for unit disk graphs as a~corollary of \Cref{thm:tss_npc_constant_grid}.
\begin{corollary}\label{cor:tss_npc_constant_udg}
	\TSS is \NPc even when the underlying graph is a~unit disk graph, all thresholds are at most~$2$ and the maximum degree is at most $3$.
\end{corollary}

\section{Majority Thresholds}\label{sec:majority_thresholds}
The last natural restriction of the threshold function, which is widely studied in the literature, is the case of majority thresholds, that is, for every $v\in V$ we have $t(v) = \lceil\deg(v)/2\rceil$.

Before delving into the specific graph classes discussed in this work, we first examine how the general case (i.e., when the underlying graph is unrestricted) is proven to be hard. The first proof of \NPhness in this setting is due to \cite{Peleg1996}. We give a sketch of a different proof that could be used to prove hardness in the general case as well, but it is important for us because it an also be used in the case of our graph classes. The idea is inspired by the proof of a related result concerning the inapproximability of \TSS given by \cite{Chen2009}. The idea is as follows. Start with an arbitrary instance of \TSSshort and inspect the vertices that do not have their threshold set to majority. Let $v$ be a vertex satisfying $t(v) \neq \lceil \frac{\deg(v)}{2}\rceil$. If the threshold is larger than majority it suffices to increase its degree by attaching dummy leaf vertices with threshold~$1$~to~$v$. Correctness of this step comes from the fact that there is always an optimal solution of \TSSshort that does not include leafs with threshold~$1$. On the other hand, when~$v$ has threshold smaller than majority, we need to increase it because we cannot safely decrease the degree of~$v$. We make use of a gadget that will supply~$v$ with sufficiently many active neighbors \emph{for free}. This is achieved by the \emph{cherry gadget}. A cherry gadget is a path on three vertices $g^\ell,g^m,g^r$ and is attached to a vertex via the middle vertex $g^m$ (see \Cref{fig:cherry_gadget}). To be more precise, there are two cases:
\begin{description}
	\item[Case 1] If $t(v)>\left\lceil \frac{\deg_G (v)}{2}\right\rceil$, attach $2t(v)-\deg_G(v)$ new vertices with threshold $1$ incident to $v$.
	\item[Case 2] If $t(v)< \left\lceil \frac{\deg_G (v)}{2}\right\rceil$, attach $\deg_G(v)-2t(v)$ cherry gadgets to $v$ as depicted in \Cref{fig:cherry_gadget}.
\end{description}

Note that the vertices $g^m$ in the cherry gadgets have degree $3$ and threshold $2$ and they have only one neighbor in the original graph. This implies that any solution $S$ is forced to contain a vertex from each cherry and there is always a optimal solution containing only the vertex $g^m$. It follows that the cherry gadgets attached to a vertex $v$ will indeed supply $v$ with sufficiently many active neighbors for free.

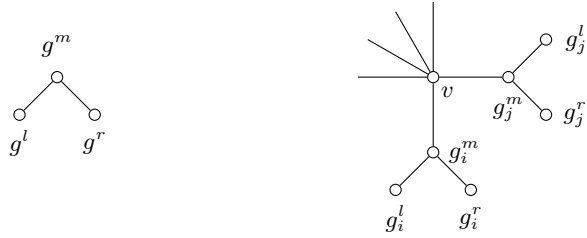
\begin{figure}[ht]
	\centering
	\begin{tikzpicture}[every node/.style={draw,circle,inner sep=1.5pt}]
	\node[label=90:{\small$g^m$}] (gm) at (-5,0){};
	\node[label=270:{\small$g^l$}] (gl) at (-5.5,-0.5){};
	\node[label=270:{\small$g^r$}] (gr) at (-4.5,-0.5){};
	
	\draw (gl) -- (gm);
	\draw (gr) -- (gm);

	\node[label=315:{\small$v$}] (v) at (0,0){};
	
	\node[label=0:{\small$g^m_i$}] (gmi) at (0,-1){};
	\node[label=270:{\small$g^l_i$}] (gli) at (-0.5,-1.5){};
	\node[label=270:{\small$g^r_i$}] (gri) at (0.5,-1.5){};
	
	\node[label=270:{\small$g^m_j$}] (gmj) at (1,0){};
	\node[label=0:{\small$g^l_j$}] (glj) at (1.5,0.5){};
	\node[label=0:{\small$g^r_j$}] (grj) at (1.5,-0.5){};
	
	\draw(v) -- (gmi);
	\draw(v) -- (gmj);
	\draw(gmi) -- (gli);
	\draw(gmi) -- (gri);
	\draw(gmj) -- (glj);
	\draw(gmj) -- (grj);
	
	\draw(v) -- (-1,0);
	\draw(v) -- (0,1);
	\draw(v) -- (-0.87,0.5);
	\draw(v) -- (-0.5,0.87);
	
	\end{tikzpicture}
	\caption{The cherry gadget (on the left). Connection of two cherry gadgets to a~vertex $v\in V(G)$ with original degree $\deg_G(v)=4$ and original threshold $t(v)=1$. The new threshold of $v$ is $t'(v)=t(v)+\deg_G(v)-2t(v)=3$ and $\deg_{G'}(v)=6$, thus it is at majority. The half edges going from~$v$ represent connection of~$v$ to the rest of~$G$.}
	\label{fig:cherry_gadget}
\end{figure}

We now utilize these observations to show \NPhness of \TSS under the majority threshold setting for our desired graph classes. We start with the planar graphs.

\begin{corollary}\label{cor:tss_npc_majority_planar}
	\TSS is \NPc under the majority threshold setting even when the underlying graph is planar with maximum degree $\Delta G \leq 4$.
\end{corollary}
\begin{proof}
	The reduction is the same as in proof \Cref{thm:tss:planar:npc:const}. We use similar gadgets but we adjust vertices whose threshold is not at majority by use of additional leaf vertices or cherry gadgets as described above. More specifically, the problematic vertices are the vertices $w_C$ in the clause gadgets. Recall that $t(w_C)=1$. If $\deg w_C \leq 2$, threshold is at majority. If $\deg w_C = 3$, we attach one cherry gadget to $w_C$.
	
	In the variable gadgets $\operatorname{VG}(x)$ vertices without majority thresholds are the vertices $F_x$ -- here attaching a leaf suffices and vertices $t_x'$ and $d_x$ -- here we attach one cherry to each. The rest of the proof is similar to the proof of \Cref{thm:tss:planar:npc:const}. Note that the cherry gadgets inevitably increased the maximum degree to~$4$.
\end{proof}

It is now straightforward to prove the \NPhness for the majority setting in the remaining graph classes. That is, grid graphs and unit disk graphs. We employ the same idea as in the proof of \Cref{thm:tss_npc_constant_grid}.

\begin{corollary}\label{cor:tss_npc_majority_grid}
	\TSS is \NPc under the majority threshold setting even if the underlying graph is a grid graph.
\end{corollary}
\begin{proof}
	Apply the same reduction as in the proof of \Cref{thm:tss_npc_constant_grid}, but start from a planar instance with majority thresholds and $\Delta G \leq 4$ which is \NPh by \Cref{cor:tss_npc_majority_planar}. Observe that a~vertex created by subdividing an edge has degree $2$ and all other degrees are unchanged. Notice that since the thresholds of the vertices created by the subdivision is $1$, the new threshold function is indeed majority.
\end{proof}

As the class of grid graphs is a~subclass of unit disk graphs (\Cref{lem:grids_are_udg}), we also obtain \NPhness under the majority setting in unit disk graphs.

\begin{corollary}\label{cor:tss_npc_majority_udg}
	\TSS is \NPc under the majority threshold setting even if the underlying graph is a unit disk graph and $\Delta G \leq 4$.
\end{corollary}

\begin{remark}\label{remark:majority3}
	Note that unless $\P=\NP$ the maximum degree $\Delta G \leq 4$ cannot be reduced to $3$ in the \NPhness of \TSSshort in the majority threshold setting even in general graphs. Observe that for $\Delta G \leq 3$ and $t(v)= \lceil \frac{\deg(v)}{2}\rceil$ we can always remove vertices with $\deg(v)\leq 1$ (there is always an optimal solution not containing these) and obtain an instance with $t(v)=\deg(v)-1$ which is equivalent to the \textsc{Feedback Vertex Set} problem. However, this problem is solvable in polynomial time in subcubic graphs~\citep{TakaokaTU2012}.
\end{remark}

\section{Exact thresholds}\label{sec:exact_thresholds}
In previous sections, we established \NPhness of \TSS in all commonly studied restrictions of the threshold function -- constant, unanimous, and majority in the classes of unit disk and planar graphs. Our proofs provided \NPhness not only for concrete graph classes but also for general graphs with very small degree. To sum up, we have hardness result for \TSS when the maximum degree is $3$ and thresholds are at most $2$ and the graph is a grid graph. Were the thresholds exactly $2$ and the maximum degree $3$, the problem is polynomial-time solvable as shown by~\cite{KynclLV2017}. In their work, they also produce a reduction showing \NPhness in the case when thresholds are \emph{exactly} $2$ and maximum degree $4$. In this section, we show that this result can indeed by achieved even in the class of unit disk graphs. For all other values of $c > 2$ the hardness of setting with $t(v) = c$ is in fact equivalent to the hardness of $c$-regular \IS or \VC, thus \NPh in general graphs.

We have a slightly weaker result for the class of unit disk graphs when the thresholds are \emph{exactly} $c$. We rely on the result about \NPhness of \IS on regular unit disk graphs (\Cref{thm:is:npc:regular,cor:inf}).

Let us start with a simple lemma.

\begin{lemma}\label{lem:reduction_2}
	Let $(G,t,k)$ be an instance of \TSS. Then there is an equivalent instance $(G',t',k')$ with $t'(v)\leq \deg_{G'}(v)$ for all $v\in V(G')$.
\end{lemma}
\begin{proof}
	If $v$ is a~vertex with threshold $t(v)>\deg (v)$, then it must be included in any target set. We thus set $G'=G-v$, decrease the threshold value of all neighbors of $v$ by $1$ (if not already at zero) and $k'=k-1$. Certainly the new instance is equivalent to $(G,t,k)$. Repeat this step until there are no vertices with threshold $t(v)>\deg (v)$.
\end{proof}

\begin{theorem}\label{thm:exact_udg}
	For infinitely many constants $c$ \TSS is \NPc when restricted to the class of unit disk graphs and the thresholds are exactly $c$. In particular, the claim holds for \mbox{$c=2,3,4$}. Moreover, for $c=2$ the hardness holds even if maximum degree of the graph is at most~$4$. 
\end{theorem}
\begin{proof}
	For $c\geq 3$ we reduce from \IS restricted to instances where the underlying graph is $c$-regular unit disk graph where $c>0$ and $c\equiv -1 \mod 4$ or $c\equiv -1 \mod 5$. \NPhness of this setting is implied by \Cref{cor:inf}. The same hardness holds for the \VC problem. By using \Cref{lem:tss:unanimous:vc:eq} the hardness for \TSS follows.
	
	For $c=2$ we reduce from \TSS with majority thresholds on grid graphs. \NPhness of this setting is implied by \Cref{cor:tss_npc_majority_grid}. Let $(G,t,k)$ be such instance and let $V(G)=\{v_1,\dots,v_n\}$. We create a~new instance $(G',t',k')$ as follows. We are aiming at $t'(v)=2$ for all $v\in V(G')$.
	
	First, we obtain a disk representation $\mathcal{D}=\{D_1,\ldots, D_n\}$ for $G$ as in proof of \Cref{lem:grids_are_udg}. Recall that grid graphs are \emph{induced} subgraphs of a grid. Recall that the disks have diameter $1$ and the center of the disk $D_i$ corresponds to the grid point of the vertex $v_i$. Now, we fix vertices~$v_i$ with threshold~$1$ by attaching a~leaf vertex~$v_i'$ with threshold~$2$ to $v_i$ and we increase the threshold of $v_i$ by $1$. In this way we have $t'(v_i)=t'(v_i')=2$. Let $z$ denote the number of vertices $v_i\in V(G)$ with $t(v_i) = 1$. We set $k'=k+z$. Let $G'$ be the newly created graph.
	
	\begin{claim*}
		The instances $(G,t,k)$ and $(G',t',k')$ are equivalent.
	\end{claim*}
	\begin{proof}
		Observe that by applying \Cref{lem:reduction_2} to the instance $(G',t',k')$ we obtain precisely the instance $(G,t,k)$. In the other direction, it is sufficient to add all the leafs $v_i'$ with threshold $2$. The claim follows.
	\end{proof}
	
	It remains to say how to realize the attachment of a~leaf vertex in the unit disk representation. Let~$s_i\in\mathbb{R}^2$ be the center of~$D_i$ and let $\varepsilon=\frac{1}{5}$. Observe that the representation satisfies: Every two disks have at most~$1$ point in common. Let $v_i$ satisfy $t(v_i)=1$. Thus, $\deg_G v_i \in \{1,2\}$ because~$t$~is majority. As all disks are embedded in an integer grid and $\deg (v_i) \leq 2$, there exists a~direction $d_i\in \{(0,1),(1,0),(-1,0),(0,-1)\}$ such that $s_i+d_i$ is not a~center of any other disk $D_j$. We add a~new disk~$D_i'$ with diameter $1$ centered at $s_i+\varepsilon d_i$ (see~\Cref{fig:udg_adding_leafs}).
	
	It is not hard to see that $D_i'\cap D_i\neq \emptyset$ and that $D_i'$ does not intersect any other disks. In other words, this exactly corresponds to attaching a~leaf vertex $v_i'$ to $v_i$. We repeat this step for all other vertices $v\in V(G)$ satisfying $t(v)=1$. Observe that the selection\footnote{The only problem might arise in the scenario shown in \Cref{fig:udg_adding_leafs}. In this case $s_i+d_i=s_j+d_j$ and the two red disks might overlap if $\varepsilon$ was chosen too large. In fact, one can compute that any $\varepsilon \in \left(0,1-\frac{1}{\sqrt{2}}\right)$ would suffice.} $\varepsilon=\frac{1}{5}$ ensures that no matter which direction~$d_j$ we choose for any other disk $D_j$, the newly created disk $D_i'$ intersects only the disk $D_i$. This can be checked directly by computing the distances of the centers of the corresponding disks.
	
	To conclude the proof, note that we added at most $1$ new vertex per each original vertex, thus the reduction is indeed polynomial.
\end{proof}

\begin{figure}[ht]
	\centering
	\begin{tikzpicture}
	\draw (0,0) circle (0.5);
	\draw (0,1) circle (0.5);
	\draw (0,-1) circle (0.5);
	\draw (1,1) circle (0.5);
	\draw (1,-1) circle (0.5);
	\draw (-1,-1) circle (0.5);
	\draw (-1,1) circle (0.5);
	\draw[red] (0+0.2,0) circle (0.5);
	\draw[red] (1,0.8) circle (0.5);
	\node[red] (diprime) at (0.2,0) {$D_i'$};
	\node[red] (djprime) at (1.8,0.6) {$D_j'$};
	\node[] (di) at (-0.8,0) {$D_i$};
	\node[] (dj) at (1.8,1) {$D_j$};
	\end{tikzpicture}
	\caption{Attaching leafs to vertices $v_i,v_j$ corresponding to disks $D_i,D_j$. The black disks correspond to the original graph while the two red disks are the newly created leafs. The corresponding directions are $d_i=(1,0)$ and $d_j=(0,-1)$.}
	\label{fig:udg_adding_leafs}
\end{figure}
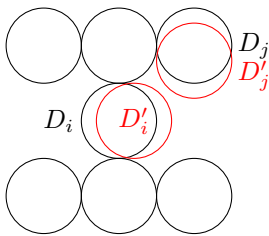

Notice that very same proof works even for planar graphs.

\begin{corollary}
	\TSS is \NPh even when the underlying graph is planar or unit disk graph with $\Delta G \leq 4$ and the thresholds are exactly $2$.
\end{corollary}

\begin{remark*}
	We remark that the first nontrivial constant for which we do not have an \NPhness result for \TSS in the class of unit disk graphs and thresholds set exactly to $c$, is $c=5$. It is implied by the fact that the proof is based heavily on the \NPhness of \IS, for which the situation is pretty much the same. Refer back to~\Cref{remark:is} for more details.
\end{remark*}

%
%
\section{Conclusions}\label{sec:conclusions}
In this paper, we showed that \TSS is computationally hard even on very simple geometric graph classes such as unit disk graphs and grid graphs. We completed the complexity picture in all commonly studied settings of the threshold function -- unanimous, constant and majority. We showed that \TSS is \NPh in the class of planar graphs or unit disk graphs when thresholds are at most $2$ and maximum degree is at most $3$ or thresholds are exactly $2$ and maximum degree is at most $4$. These hardness results are tight in the sense that further restrictions of the maximum degree or the threshold function yields polynomial-time solvable instances.

As a corollary to hardness in planar graphs with small degree, we showed that \TSS is \NPh even in the class of grid graphs when thresholds are at most $2$ and maximum degree at most $3$. We then utilized these results to show \NPhness of \TSS in the majority threshold setting in grid graphs and unit disk graphs with $\Delta G \leq 4$. As noted in \Cref{remark:majority3}, further reducing the maximum degree to $3$ is not possible unless $\P=\NP$.

Throughout the reductions we
utilized the so-called rectilinear embedding which might be useful in reductions of other problems on grid graphs.

\subsection*{Future directions and open questions}
The first obvious question is, whether the setting $t(v)=2$ is also \NPh in the class of grid graphs. The approach used in \Cref{thm:exact_udg} is not applicable in grid graphs. It might work if the subgraph of the grid need not be induced, however one would need to be careful where to place the leafs.

\begin{question*}
	What is the complexity of \TSS in the class of grid graphs with $t(v)=2$?
\end{question*}

If even this setting turns out to be hard, we would still like to know whether there is some reasonable restriction of the structure of the grid graph or the threshold function that makes \TSSshort tractable in grid graphs.

Another question that remains open even after this paper is related to the computational complexity of the \TSS problem on interval graphs. It was known that \TSSshort is polynomial-time solvable on interval graphs when the threshold function is bounded by a constant. As a corollary to tractability of \VC in interval graphs, the unanimous threshold setting is also polynomial-time solvable in this class. We conjecture that the problem should be tractable with majority thresholds; however, we are more skeptical in the case of general thresholds.

\begin{question*}
	What is the complexity of \TSS on the class of interval graphs in the majority and unrestricted threshold setting?
\end{question*}

Finally, an interesting open problem, which is slightly unrelated to \TSS is about the \NPhness of the \IS problem in $r$-regular unit disk graphs for all constants $r\geq 5$. In this work, we established \NPhness for $r=3,4$ and infinitely many (but not all) constants (refer to \Cref{remark:is} for details).

\begin{question*}
	Is \IS \NPh when restricted to the class of $r$-regular unit disk graphs for all constants $r\geq 5$?
\end{question*}

Similar question follows with the \NPhness of exact threshold setting in the class of unit disk graphs as our proof relies on the result about \IS.

\begin{question*}
	Given any $c\geq 2$, is \TSS \NPh even when restricted to the class of unit disk graphs and all thresholds are exactly $c$?
\end{question*}

\subsubsection*{Acknowledgments.}
The authors are grateful to the reviewer’s valuable comments that improved the manuscript.
The authors acknowledge the support of the Czech Science Foundation Grant No. 22-19557S. Michal Dvořák and Šimon Schierreich were additionally supported by the Grant Agency of the Czech Technical University in Prague, grant No. SGS23/205/OHK3/3T/18. Michal Dvořák was also supported by the Student Summer Research Program 2021 of FIT CTU in Prague.

\FloatBarrier
\bibliographystyle{abbrvnat}
\bibliography{references}

\end{document}